\RequirePackage{fix-cm}
\documentclass[smallextended]{svjour3}       
\smartqed  

\usepackage{times}
\usepackage{algorithmic}
\usepackage{algorithm}
\usepackage{epsfig}
\usepackage{amssymb}
\usepackage{amsmath}
\usepackage{amsfonts}
\usepackage{subfigure}
\usepackage{nccmath} 
\usepackage{graphicx}
\usepackage{url}
\usepackage{multirow}
\usepackage{wrapfig}
\usepackage{subfigure}
\usepackage{caption}
\usepackage{color}
\usepackage{tabularx}
\usepackage{sidecap}
\usepackage{varwidth}
\usepackage[usenames,dvipsnames]{xcolor}

\begin{document}

\title{Mining Rooted Ordered Trees under Subtree Homeomorphism}


\author{Mostafa Haghir Chehreghani  \and \\ Maurice Bruynooghe}
\institute{
Mostafa Haghir Chehreghani \and Maurice Bruynooghe 
\at
	      Department of Computer Science, KU Leuven, 3001 Leuven, Belgium \\
Mostafa.HaghirChehreghani@cs.kuleuven.be, Maurice.Bruynooghe@cs.kuleuven.be 
}

\date{Received: date / Accepted: date}

\maketitle

\begin{abstract}
Mining frequent tree patterns has many applications in different areas such as
XML data, bioinformatics and World Wide Web.
The crucial step in frequent pattern mining is frequency counting, which involves
a matching operator to find occurrences (instances) of a tree pattern in a given collection of trees.
A widely used matching operator for tree-structured data is \textit{subtree homeomorphism},
where an edge in the tree pattern is mapped onto an ancestor-descendant relationship in the given tree.
Tree patterns that are frequent under subtree homeomorphism are usually called \textit{embedded patterns}. 
In this paper, we present an efficient algorithm for subtree homeomorphism with application to frequent pattern mining.
We propose a compact data-structure, called \textbf{occ}, which stores only information about the
rightmost paths of occurrences and hence can encode and represent several occurrences of a tree pattern.
We then define efficient join operations on the \textbf{occ} data-structure,
which help us count occurrences of tree patterns according to occurrences of their proper subtrees.
Based on the proposed subtree homeomorphism method, we develop an effective pattern mining algorithm, called TPMiner.
We evaluate the efficiency of TPMiner on several real-world and synthetic datasets.
Our extensive experiments confirm that TPMiner always outperforms well-known existing algorithms,
and in several cases the improvement with respect to  existing algorithms is significant.

\keywords{XML data \and rooted ordered trees \and frequent tree patterns \and subtree homeomorphism \and embedded subtrees}
\end{abstract}

\section{Introduction}
\label{intro}

Many semi-structured data such as XML documents are represented by rooted ordered trees.
One of the most important problems in the data mining of these tree-structured data 
is frequent tree pattern discovery.
Mining frequent tree patterns is very useful in various domains such as network routing \cite{8},
bioinformatics \cite{36} and user web log data analysis \cite{Ivancsy}.
Furthermore, it is a crucial step in several other data mining and machine learning problems such as clustering and classification \cite{37}.


In general, algorithms proposed for finding frequent tree patterns include two main phases:
1) generating candidate tree patterns, and
2) counting the frequency of every generated tree pattern in
a given collection of trees (called the \emph{database trees} from now on).
The generation step (which involves a refinement operator) is computationally easy.
There are methods, such as \textit{rightmost path extension},
that can generate efficiently all non-redundant rooted ordered trees, i.e., each in $O(1)$ computational time \cite{36} and \cite{2}.
The frequency counting step is computationally expensive.
Empirical comparison of these two phases can be found e.g., in \cite{4},
where Chi et. al. showed that a significant part of the time required for finding frequent patterns
is spent on frequency counting. 
Thereby, the particular method used for frequency counting can significantly affect the efficiency of the tree mining algorithm.

The frequency counting step involves a \textit{matching operator} \cite{12}.
A widely used matching operator between a tree pattern and a database tree is \textit{subtree homeomorphism};
an injective mapping of the vertices such that an edge in the tree pattern is mapped onto an ancestor-descendant relationship in the database tree.
Frequent tree patterns under subtree homeomorphism are called \textit{embedded patterns}.
They have many applications in different areas.
Zaki et. al. \cite{37} presented XRules, a classifier based on frequent embedded tree patterns, and
showed its high performance compared to classifiers such as SVM.
Ivancsy and Vajk used frequent embedded tree patterns for analyzing the navigational behavior of the web users \cite{Ivancsy}.

Two widely used frequency notions are \textit{per-tree frequency},
where only the occurrence of a tree pattern inside a database tree is important;  
and \textit{per-occurrence frequency}, where the number of occurrences is important, too.
While there exist algorithms optimized for the first notion \cite{36}, \cite{24} and \cite{26}, 
this notion is covered also by the algorithms proposed for the second notion.
Per-occurrence frequency counting is computationally more expensive than per-tree frequency counting.
In the current paper, our concern is \textit{per-occurrence frequency}.
An extensive discussion about the applications in which the second notion is preferred can be found e.g., in \cite{23}.
One of the investigated examples is a digital library where
author information are separately stored in database trees in some form, e.g., \textsf{author--book--area--publisher}.
A user may be interested in finding out information about the popular \textsf{publisher}s of every \textsf{area}.
Then, the repetition of items within a database tree becomes important,
hence, per-occurrence frequency is more suitable than per-tree frequency \cite{23}.

Two categories of approaches have been used for counting occurrences of tree patterns under subtree homeomorphism.
The first category includes approaches that use one of the algorithms proposed for subtree homeomorphism between two trees.
HTreeMiner \cite{36} employs such an approach.
The second category includes approaches that store the information representing/encoding the occurrences of tree patterns.
Then, when the tree pattern is extended to a larger one,
its stored information is also extended, in a specific way,
to represent the occurrences of the extended pattern.
These approaches are sometimes called \textit{vertical} approaches.
VTreeMiner \cite{36} and MB3Miner \cite{23} are examples of the methods that use such vertical approaches. 
As studied in \cite{36}, vertical approaches are more efficient than the approaches in the first category.
Many efficient vertical algorithms are based on the numbering scheme proposed by Dietz \cite{Dietz}.
This scheme uses a tree traversal order to determine the ancestor-descendant relationship between pairs of vertices.
It associates each vertex with a pair of numbers, sometimes called \textit{scope}.
For instance, VTreeMiner and TreeMinerD \cite{36} and TwigList \cite{Qin}
use this scheme in different forms, to design efficient methods for counting occurrences of tree patterns.

The problem with these algorithms is that in order to count all occurrences,
they use data-structures that represent whole occurrences.
This renders the algorithms inefficient,
especially when patterns are large and have many occurrences in the database trees.
In the worst case,
the number of occurrences of a tree pattern can be exponential in terms of the size of pattern and database \cite{Chisurvey}.
Therefore, keeping track of all occurrences can significantly reduce the efficiency
of the algorithm, in particular when tree patterns have many occurrences in the database trees.

The main contribution of the current paper is to introduce a novel vertical algorithm for the class of rooted ordered trees.
It uses  a more compact data-structure, called \textbf{occ} (an abbreviation for \textbf{oc}currence \textbf{c}ompressor)
for representing occurrences,
and a more efficient subtree homeomorphism algorithm based on Dietz's numbering scheme \cite{Dietz}.
An \textbf{occ} data-structure stores only information about rightmost paths of
occurrences and hence can represent/encode all occurrences that have
the rightmost path in common.
The number of such occurrences can be exponential,
even though the size of the \textbf{occ} is only $O(d)$, where $d$ is the length of the rightmost path of the tree pattern.
We present efficient join operations on \textbf{occ} that help us to
efficiently calculate the occurrence count of tree patterns from the
occurrence count of their proper subtrees.
Furthermore, we observed that in most of widely used real-world databases,
while many vertices of a database tree have the same label,
no two vertices on the same path are identically labeled.
For this class of database trees, worst case space complexity of our algorithm is linear;
a result comparable to the best existing results for per-tree frequency.
We note that for such databases, worst case space complexity of
the well-known existing algorithms for per-occurrence frequency,
such as VTreeMiner \cite{36} and MB3Miner \cite{23}, is still exponential.
Based on the proposed subtree homeomorphism method, we develop
an efficient pattern mining algorithm, called TPMiner.
To evaluate the efficiency of TPMiner, we perform extensive experiments on both real-world and synthetic datasets.
Our results show that TPMiner always outperforms most efficient existing algorithms such as VTreeMiner \cite{36} and MB3Miner \cite{23}.
Furthermore, there are several cases where the improvement of TPMiner with respect to
existing algorithm is very significant.

In Section~\ref{sec:preliminaries}, preliminaries and definitions related to the tree pattern mining problem are introduced.
In Section~\ref{sec:relatedwork}, a brief overview on related work is given.
In Section~\ref{sec:homeomorphism}, we present the \textbf{occ} data-structure and our subtree homeomorphism algorithm.
In Section~\ref{sec:tpminer}, we introduce the TPMiner algorithm for finding frequent embedded tree patterns from rooted ordered trees.
We empirically evaluate the effectiveness of TPMiner in Section~\ref{sec:experimentalresults}.
Finally, the paper is concluded in Section~\ref{section:conclusion}.

\section{Preliminaries}
\label{sec:preliminaries}

We assume the reader is familiar with the basic concepts in graph theory.
The interested reader can refer to e.g., \cite{Diestel}.
An \textit{undirected (vertex-labeled) graph} $G=(V,E,\lambda)$ consists of a vertex set $V$,
an edge set $E\subseteq \{e\subseteq V : |e|=2 \}$, and a labeling function $\lambda$: $V \to \Sigma$ which assigns
a label from a finite set $\Sigma$ to every vertex in $V$.
$G=(V,E,\lambda)$ is a \textit{directed graph} if $E$ is: $E\subseteq V\times V$, where $\times$ is the \textit{Cartesian product}.
We use the notations $V(G)$, $E(G)$ and $\lambda_G$ to refer to the set of
vertices, the set of edges (or arcs) and the labeling function of $G$, respectively. 
The \textit{size} of $G$ is defined as the number of vertices of $G$.
Two graphs $G_1=(V_1,E_1,\lambda_1)$ and $G_2=(V_2,E_2,\lambda_2)$ (either both are directed or both are undirected) are identical,
written as $G_1 = G_2$, if $V_1=V_2$, $E_1=E_2$, and $\forall v \in V_1: \lambda_1(v)=\lambda_2(v)$.
A \textit{path} from a vertex $v_0$ to a vertex $v_n$ in a directed graph $G=(V,E,\lambda)$
is a sequence of vertices such that $\forall i$, $0 \leq i \leq n-1$, $(v_i,v_{i+1}) \in E(G)$.
The \textit{length} of a path is defined as its number of edges (number of vertices minus 1).
A \textit{cycle} is a path with $v_0 = v_n$.

\begin{figure*}
\centering
\includegraphics[scale=0.85]{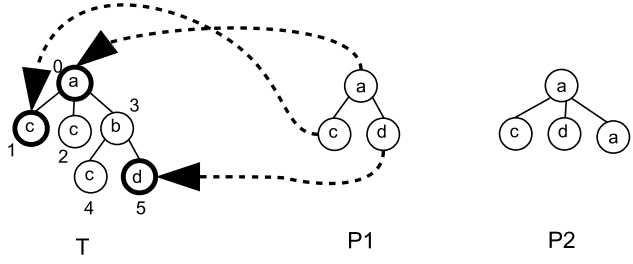}
\caption{
$T$ is a database tree where numbers next to vertices present preorder numbers;
$P1$ is a tree pattern where the vertices labeled by $a$ and $d$ form its rightmost path;
$P2$ is a rightmost path extension of $P1$
where a new vertex labeled by $a$ is added to an existing vertex of $P1$ labeled by $a$ as the rightmost child.
Finally, dashed lines present a subtree homeomorphism mapping from $P1$ to $T$.}
\label{fig:basicexample} 
\end{figure*}

An undirected graph not containing any cycles is called a \textit{forest} and a connected forest is called a (\textit{free}) \textit{tree}.
A \textit{rooted tree} is a directed acyclic graph (DAG) in which:
1) there is a distinguished vertex, called \textit{root},
that has no incoming edges, 2) every other vertex has exactly one incoming edge,
and 3) there is an unique path from the root to any other vertex.
In a rooted tree $T$, $u$ is the \textit{parent} of $v$ ($v$ is the \textit{child} of $u$) if $(u,v) \in E(T)$.
The transitive closure of the parent-child relation is called the \textit{ancestor-descendant} relation.
A \textit{rooted ordered tree} is a rooted tree such that there is an order
over the children of every vertex.
Throughout this paper, we refer to \textit{rooted ordered trees} simply as \textit{trees}.
\textit{Preorder traversal} of a tree $T$ is defined recursively as follows:
first, visit $root(T)$; and then for every child $c$ of $root(T)$ from left to right,
perform a preorder traversal on the subtree rooted at $c$.
The position of a vertex in the list of visited vertices during a preorder traversal is called its \textit{preorder number}.
We use $p(v)$ to refer to the preorder number of vertex $v$.
For instance, in Figure \ref{fig:basicexample},
numbers next to vertices of tree $T$ present their preorder numbers.
The \textit{rightmost path} of $T$ is the path from $root(T)$ to the last vertex of $T$ visited in the preorder traversal. 
For example, in Figure \ref{fig:basicexample}, the vertices labeled by $a$ and $d$ form the rightmost path of $P1$.
Two distinct vertices $u$ and $v$ are \textit{relatives}
if $u$ is neither an ancestor nor a descendant of $v$.
With $p(u) < p(v)$, $u$ is a left relative of $v$, otherwise, it is a right relative.
For example, in tree $T$ of Figure \ref{fig:basicexample}, vertices $1$ and $4$ are relatives
and vertex $4$ is a right relative of vertex $1$.

A tree $T$ is a \textit{rightmost path extension} of a tree $T'$ iff there exist
vertices $u$ and $v$ such that:
(i) $\{v\} = V(T) \setminus V(T')$,
(ii) $\{(u,v)\} = E(T) \setminus E(T')$, 
(iii) $u$ is on the rightmost path of $T'$, and
(iv) in $T$, $v$ is a right relative of all children of $u$.
We say that $T$ is the rightmost path extension of $T'$ with $v$ attached
at $u$ and we denote $T$ as $RExtend(T',v,u)$.
For instance, in Figure \ref{fig:basicexample}, $P2$ is a rightmost path extension of $P1$,
where the new vertex $v$ is labeled by $a$,
the existing vertex $u$ is also labeled by $a$,
and $v$ is added to $u$ as the rightmost child.

A tree $P$ is \textit{subtree homeomorphic} to a tree $T$ (denoted by $P \preceq_h T$)
iff there is a mapping $\varphi:V(P)\to V(T)$ such that:
(i) $\forall v \in V(P): \lambda_P(v) = \lambda_T(\varphi(v))$,
(ii) $\forall u,v \in V(P)$: $u$ is the parent of $v$ in $P$ iff $\varphi(u) \text{ is an ancestor of $\varphi(v)$ in $T$}$, and
(iii) $\forall u,v \in V(P)$: $p(u)<p(v) \Leftrightarrow p(\varphi(u)) < p(\varphi(v))$.
For example, in Figure \ref{fig:basicexample}, dashed lines present a subtree homeomorphism mapping
from $P1$ to $T$.
Under subtree isomorphism and homomorphism, 
the ancestor-descendant relationship between $\varphi(u)$ and $\varphi(v)$ in (ii) is strengthened into
the parent-child relationship;
under subtree homomorphism, (iii) is weakened into:
$p(u)<p(v) \Leftrightarrow p(\varphi(u)) \leq p(\varphi(v))$, i.e., 
successive children of a vertex in
the pattern can be mapped onto the same vertex in the database tree.
$P$ is \textit{isomorphic} to $T$ (denoted $P \cong_i T$)
iff $P$ is subtree isomorphic to $T$ and $|V(P)| = |V(T)|$.

Note that a pattern $P$ under a subtree morphism can have several mappings
to the same database tree $T$.
When the matching operator is subtree homeomorphism,
every mapping is called an \textit{occurrence} (or \textit{embedding}) of $P$ in $T$. 
An occurrence (embedding) of a vertex $v$
is an occurrence (embedding) of the pattern consisting of the single vertex $v$.
The number of occurrences of $P$ in $T$ is denoted by $NumOcc(P,T)$.

Given a database $\mathcal{D}$ consisting of trees and a tree $P$,
\textit{per-tree support} (or \textit{per-tree frequency}) of $P$ in $\mathcal{D}$ is defined as:
$|\{T \in \mathcal{D} : P \preceq_h T \}|$.
\textit{Per-occurrence support} (or \textit{per-occurrence frequency}) of $P$ in $\mathcal{D}$ is defined as:
$\sum_{T\in \mathcal{D}} NumOcc(P,T)$.
In this paper, our focus is \textit{per-occurrence support}.
For the sake of simplicity, we use the term \textit{support} (or \textit{frequency})
instead of \textit{per-occurrence support} (or \textit{per-occurrence  frequency}), and denote it by $sup(P,\mathcal{D})$.
$P$ is \textit{frequent} ($P$ is a \textit{frequent embedded pattern}),
iff its support is greater than or equal to a user defined integer threshold $minsup > 0$. 
The problem studied in this paper is as follows:
given a database $\mathcal{D}$ consisting of trees and an integer $minsup$,
find every frequent pattern $P$ such that $sup(P,\mathcal{D}) \geq minsup$.

\begin{figure*}
\centering
\includegraphics[scale=0.6]{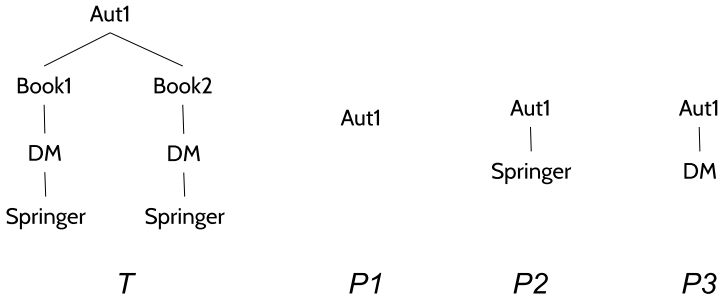}
\caption{
In the database tree $T$ and for $minsup=2$, while $P1$ is infrequent, it has two frequent supertrees $P2$ and $P3$.
}
\label{fig:antimonotone} 
\end{figure*}

We observe that when \textit{per-occurrence support} is used,
anti-monotonicity might be violated: 
it is possible that the support of $P$ is greater than or equal to $minsup$,
but it has a subtree whose support is less than $minsup$.
For example, consider the database tree $T$ of Figure \ref{fig:antimonotone}
and suppose that $minsup$ is $2$.
Then, while the pattern $P1$ is infrequent,
it has two frequent supertrees $P2$ and $P3$.
Therefore, in a more precise (and practical) definition, which is also used by algorithms
such as VTreeMiner \cite{36},
tree $P$ is frequent iff: 1) $sup(P,\mathcal{D}) \geq minsup$,
and 2) the subtree $P'$ generated by removing the rightmost vertex of $P$ is frequent.
This means only frequent trees are extended to generate larger patterns.


\section{Related work}
\label{sec:relatedwork}

Recently, many algorithms have been proposed for finding frequent embedded patterns from a database of tree-structured data
that work with both \textit{per-tree} and \textit{per-occurrence} frequencies.
Zaki presented VTreeMiner \cite{36} to find embedded patterns from trees.
For frequency counting he used an efficient data structure, called \textit{scope-list}, and proposed rightmost path extension
to generate non-redundant candidates.
Later, he proposed the SLEUTH algorithm to mine embedded patterns from rooted unordered trees \cite{35}.
Tan et. al. \cite{23} introduced the MB3Miner algorithm,
where they use a unique occurrence list representation of the tree structure,
that enables efficient implementation of their Tree Model Guided (TMG) candidate generation.
TMG can enumerate all the valid candidates that fit in the structural aspects of the database.
Chaoji et. al. \cite{Chaoji} introduced a generic pattern mining approach that supports various types of patterns and frequency notions. 
A drawback of these algorithms is that in order to count the number of occurrences 
of a tree pattern $P$ in a database tree $T$, they need to keep track of all occurrences of $P$ in $T$.
For example, in VTreeMiner, for every occurrence $\varphi$ of $P$ in $T$ a separate element is stored in \textit{scope-list}, 
that consists of the following components:
($i$) $TId$ which is the identifier of the database tree that contains the occurrence,
($ii$) $m$ which is $\{\varphi(v) | v \in V(P) \setminus \{\text{rightmost vertex of $P$}\}\}$, and
($iii$) $s$ which is the scope of $\varphi(u)$, where $u$ is the rightmost vertex of $P$.
In the current paper, we propose a much more compact data-structure that can represent/encode all occurrences that have
the rightmost path in common in $O(d)$ space, where $d$ is the length of the rightmost path of $P$.
The number of such occurrences can be exponential.
We then present efficient algorithms that
calculate the occurrence count of tree patterns from the occurrence count
of their proper subtrees.

There also exist several algorithms optimized for \textit{per-tree frequency}.
Zaki presented TreeMinerD \cite{36} based on the \textit{SV-list} data-structure.
Tatikonda et. al. \cite{24} developed TRIPS and TIDES using two sequential encodings of trees
to systematically generate and evaluate tree patterns.
Wang et. al. \cite{26} proposed XSpanner that uses a pattern growth-based method to find embedded tree patterns.
The literature contains also many algorithms for mining patterns
under subtree isomorphism; examples are in \cite{2}, \cite{18}, \cite{Balcazar}, \cite{Chehreghani1} and \cite{Chehreghani2}.
From the applicability point of view, for instance in the prediction task,
when longer range relationships are relevant, 
subtree homeomorphism becomes more useful than subtree isomorphism.
For a more comprehensive study of related work,
the interested reader can refer to e.g., \cite{Chisurvey}.

A slightly different problem over rooted, ordered, and labeled trees
is the tree inclusion problem: can a pattern $P$ be obtained from a tree
$T$ by deleting vertices from $T$. In our terminology, is $P$ present in $T$
under subtree homeomorphism? Bille and Gortz \cite{Bille} recently proposed a
novel algorithm that runs in linear space and subquadratic time,
improving upon a series of polynomial time algorithms that started
with the work in \cite{12}.

\section{Efficient tree mining under subtree homeomorphism}
\label{sec:homeomorphism}

In this section, we present our method for subtree homeomorphism of rooted ordered trees.
First in Section~\ref{sec:occurrencetrees},
we introduce the notion of \textit{occurrence tree} and its rightmost path extension.
Then in Section~\ref{sec:occ}, we present the \textbf{occ-list} data-structure and, in Section~\ref{sec:join},
the operators for this data structure.
In Section~\ref{sec:completely}, we analyze space and time complexities of our proposed frequency counting method. 
We briefly compare
our approach with other vertical frequency counting methods in Section~\ref{sec:briefcomparison}.

\subsection{Occurrence trees and their extensions}
\label{sec:occurrencetrees}

Under rightmost path extension, a pattern $P$ with $k + 1$ vertices is
generated from a pattern $P'$ with $k$ vertices by adding a vertex $v$, as the
rightmost child, to a vertex in the rightmost path of $P$. Occurrences
of $P$ are rightmost path extensions of occurrences of $P'$ with an
occurrence of $v$. Therefore, an interesting way to construct
occurrences of $P$ is to look at the occurrences of $v$ that can be a
rightmost path extension of an occurrence of $P'$.
First, we introduce the notion of \textit{occurrence tree}.
Then, we present the conditions under which a rightmost path extension of an occurrence tree yields another occurrence tree.

\begin{definition}[Occurrence tree]\label{def:Occurrencetree}
Given an occurrence $\varphi$ of $P$ in $T$, we define the \emph{occurrence tree} $\mathcal{OT}(\varphi)$ as follows:
(i) $V(\mathcal{OT}(\varphi))=\{ \varphi(v): v \in V(P) \}$,
(ii) $root(\mathcal{OT}(\varphi)) = \varphi(root(P))$,
for every $v \in V(\mathcal{OT}(\varphi))$, $\lambda_{\mathcal{OT}(\varphi)} (v)=\lambda_{T} (v)$, and
(iii) $E(\mathcal{OT}(\varphi)) = \{(\varphi(v_1),\varphi(v_2)) | (v_1,v_2) \in E(P)\}$.
\end{definition}

\begin{figure*}
\centering
\includegraphics[scale=0.7]{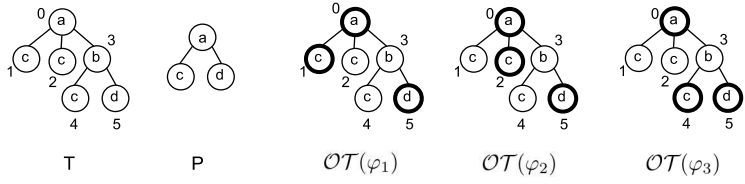}
\caption{From left to right, a database tree $T$, a pattern $P$ and three
occurrence trees $\mathcal{OT}(\varphi_1)$, $\mathcal{OT}(\varphi_2)$, and
$\mathcal{OT}(\varphi_3)$. Labels are inside vertices, preorder numbers are next
to vertices. The occurrence trees are represented by showing the
occurrences of the pattern vertices in bold. Their edges are the images of
the edges in the pattern; for example, $\mathcal{OT}(\varphi_3)$ refers to the tree
formed by vertices $0$, $4$ and $5$, with $0$ as the root, and with $(0,4)$ and
$(0,5)$ as edges.
}
\label{fig:1} 
\end{figure*}

Notice, when $(v_1,v_2) \in E(\mathcal{OT}(\varphi))$ and $v_1$ is not the parent of $v_2$ in $T$,
then all intermediate vertices on the path from $v_1$ to $v_2$ are not part of $V(\mathcal{OT}(\varphi))$.
For example, in Figure~\ref{fig:1}, the tree $P$ has $3$ occurrences in tree $T$. 

Selecting a vertex not yet in the occurrence tree and performing a
rightmost path extension does not always result in another occurrence
tree.
Proposition~\ref{prop:RPE1} below lists properties that hold when the rightmost path extension is an occurrence tree.

\begin{proposition}
\label{prop:RPE1}
Let $\varphi'$ be an occurrence of a pattern $P'$ in a database tree $T$ and $OT' = \mathcal{OT} (\varphi')$. 
Let $x$ be a vertex of $T$ outside $OT'$, and $y$ a vertex on the rightmost path of $OT'$.
If $OT=RExtend(OT',x,y)$ is an occurrence tree in $T$, then:
(i) $root(OT')$ is an ancestor of $x$ in $T$,
(ii) of all ancestors of $x$ in $T$ that belong to $OT'$, $y$ is the largest
one in the preorder over $T$, and
(iii) for each vertex $w$ in $OT'$ such that $p(w) > p(y)$, $w$ is a left relative
of $x$ in $OT$ and in $T$.
\end{proposition}

\begin{proof}
By the definition of rightmost path extension presented in Section~\ref{sec:preliminaries},
either $y$ is $root(OT')$ or it is a descendant of $root(OT')$; moreover $y$
is the parent of $x$, hence (i) holds.
As $y$ is on the rightmost path in $OT'$ and $y$ is the parent of $x$ in $OT$ (and children of $y$ are relatives of $x$ in $T$), $y$
is, among the ancestors of $x$ in $T$, the largest one that belongs to $OT'$
(ii).
If $p(w) > p(y)$ for a vertex $w \in V(OT')$, since $y$ is on the rightmost path of $OT'$,
$w$ is a descendant of $y$.
Furthermore, since $x$ is the rightmost child of $y$ in $OT$, $w$ is a left relative of $x$ in $OT$ and in $T$ (iii).
$\blacksquare$
\end{proof}

\begin{figure*}
\centering
\subfigure[]
{
\includegraphics[scale=0.75]{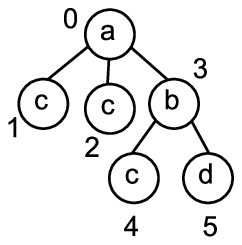}
\label{fig:3a}
}\quad
\subfigure[]
{
\includegraphics[scale=0.75]{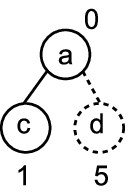}
\label{fig:3b}
}\quad
\subfigure[]
{
\includegraphics[scale=0.75]{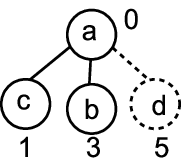}
\label{fig:3c}
}\quad
\subfigure[]
{
\includegraphics[scale=0.75]{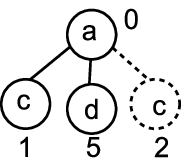}
\label{fig:3d}
}
\subfigure[]
{
\includegraphics[scale=0.75]{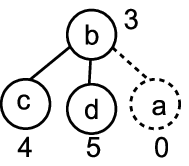}
\label{fig:3e}
}
\caption{\label{fig:3}
A database tree is shown \ref{fig:3a}, together with four rightmost path
extensions of different occurrence trees in that database tree.
However, only \ref{fig:3b} is itself an occurrence tree in the database tree;
\ref{fig:3c}, \ref{fig:3d} and \ref{fig:3e} violate conditions (ii), (iii) and (i) of Proposition \ref{prop:RPE2}, respectively.
}
\end{figure*}

As an example of Proposition \ref{prop:RPE1}, consider Figure \ref{fig:3},
where \ref{fig:3a} presents a database tree
and \ref{fig:3b} shows that 
an occurrence tree $OT'$ consisting of vertices $0$ and $1$
is extended to another occurrence tree $OT$ consisting of vertices $0$, $1$ and $5$.
Vertex $0$ is an ancestor of vertex $5$ in the database tree (condition ($i$));
among all ancestors of vertex $5$ in the database tree that belong to $OT'$,
vertex $0$ is the largest one in the preorder over the database tree (condition ($ii$));
and vertex $1$ is a left relative of vertex $5$ in $OT$ and in the database tree (condition ($iii$)).

Let $P'$ be a tree pattern.
The next proposition lists the conditions that are sufficient to ensure
that a rightmost path extension of an occurrence tree $OT(\varphi')$ of $P'$ with an edge $(y,x)$ yields an
occurrence tree of another tree pattern $P$, where $P$ is a rightmost path extension of $P'$.

\begin{proposition}
\label{prop:RPE2}
Let $\varphi'$ be an occurrence of a pattern $P'$ in a database tree $T$ and $OT' = \mathcal{OT} (\varphi')$. 
Let $x$ be a vertex of $T$ outside $OT'$, and $y$ a vertex on the rightmost path of $OT'$.
If:
(i) $y$ is an ancestor of $x$ in the database tree $T$,
(ii) of all vertices of $OT'$ that are ancestors of $x$ in the database tree $T$, $y$ is the largest one, and
(iii) $p(x) > p(w)$ for all $w \in OT'$, 
then $RExtend(OT',x,y)$ is an occurrence tree.
For a given vertex $x$ and an occurrence tree $OT$, 
if there exists a vertex $y$ such that $RExtend(OT,x,y)$,
$y$ is unique.
\end{proposition}

\begin{proof}
Let $u$ be the vertex of $P'$ such that $y=\varphi(u)$. To define $P$, we set $V(P)=
V(P') \cup \{v\}$ with $v$ a new vertex with the same label as $x$, and $E(P)
= E(P') \cup \{(u,v)\}$ such that $v$ is the rightmost child of $u$. Define
$OT$ as $V(OT) = V(OT')\cup \{x\}$ and $E(OT) = E(OT) \cup \{(y,x)\}$ and
set $\varphi(P)= \varphi'(P') \cup \{v \rightarrow x\}$.
By the construction and the assumptions, $OT=RExtend(OT',x,y)$.
We show $OT$ is a occurrence tree of $P$ in $T$.
From (i) and (ii) it follows that $x$ is on the rightmost path from $root(OT)$ and
it is a descendant of $y$ in $T$ 
and from (iii) that the children of $y$ in $OT$ are left relatives of $x$, hence,
$\varphi$ is an embedding of $P$ in $T$ and $OT$ is an occurrence tree of $P$ in $T$.
We note since all vertices of a tree have a unique preorder number, 
$y$ is unique, if it exists.
$\blacksquare$
\end{proof}


Figure~\ref{fig:3} shows a database tree (\ref{fig:3a}) and four rightmost path extensions
of occurrence trees; however, only one of these extensions is
another occurrence tree (\ref{fig:3b}), the other ones violate the conditions of
Proposition~\ref{prop:RPE2}.

To turn the conditions of Proposition~\ref{prop:RPE2} into a practical method, we
need a compact way to store occurrence trees and an efficient way to
check the conditions. For the latter, we take advantage of the solution of Dietz \cite{Dietz}.
He has designed a numbering scheme based on tree traversal order to determine
the ancestor-descendant relationship between any pair of vertices.
This scheme associates each vertex $v$ with a pair of numbers $\left< p(v), p(v)+size(v) \right>$,
where $size(v)$ is an integer with certain properties
(which are met when e.g., $size(v)$ is the number of descendants of $v$).
Then, for two vertices $u$ and $v$ in a given database tree, $u$ is an ancestor of $v$ iff $p(u) < p(v)$ and  $ p(v)+ size(v) \leq p(u) + size(u)$,
and $v$ is a right relative of $u$ iff $p(u)+size(u) < p(v)$.
In several algorithms, such as VTreeMiner and TreeMinerD \cite{36} and TwigList \cite{Qin},
variants of this scheme have been used to design efficient methods for counting occurrences
of tree patterns based on occurrences of their subtrees.

Our contribution is to introduce data structures, based on the Dietz numbering scheme,
that allow us to speed up counting of all occurrences of tree patterns.
For example, while the algorithm of \cite{36} keep track of all occurrences, we only store the occurrences that have distinct rightmost paths.
We start with introducing some additional notations.
The scope of a vertex $x$ in a database tree $T$, denoted $x.scope$, is a pair $(l, u)$,
where $l$ is the preorder number of $x$ in $T$ and $u$ is the
preorder number of the rightmost descendant of $x$ in $T$.
We use the notations $x.scope.l$ and $x.scope.u$ to refer to $l$ and $u$ of the scope of $x$.

\begin{definition}[rdepth]
Let $x$ be a vertex on the rightmost path of a tree $T$.
The $rdepth$ of $x$ in $T$, denoted $ rdep_T(x)$, is
the length of the path from the root of $T$ to $x$.
\end{definition}

A vertex $x$ on the rightmost path of $T$ is uniquely distinguished by $rdep_T(x)$.
For example in Figure~\ref{fig:3a}, the \textit{rdepth} of vertices $0$, $3$ and $5$ is $0$, $1$ and $2$, respectively
(and for the other vertices, \textit{rdepth} is undefined).

\begin{proposition}
\label{proposition:rightmostpathextensiona}
Let $OT'$ be an occurrence tree of a tree pattern $P'$ in a database tree $T$, $x \in V(T) \setminus  V(OT')$
and $y$ a vertex on the rightmost path of $OT'$ but not the rightmost vertex
(i.e., the rightmost path of $OT'$ has a vertex $z$ such that $rdep_{OT'} (z) = rdep_{OT'}(y)+ 1$).
We have: $RExtend( OT',x,y)$ is an occurrence tree iff
\begin{equation}
\label{eq:proposition2:eq0}
z.scope.u < x.scope.l \leq y.scope.u 
\end{equation}
\end{proposition}
\begin{proof}
First, assume $RExtend(OT',x,y)$ is an occurrence tree.
By Proposition~\ref{prop:RPE1}, $y$ is the largest vertex of $OT'$ that is an ancestor of $x$,
hence, in the database tree $T$, the tree rooted at $x$ is a subtree of the tree rooted at $y$.
That means
\begin{equation}
\label{eq:a}
y.scope.l < x.scope.l \leq x.scope.u \leq y.scope.u
\end{equation}
Vertex $z$ and all vertices in the subtree of $z$ are left
relatives of $x$ and have a preorder number smaller than that of $x$.
Hence $z.scope.u < x.scope.l$.
Combining with Inequation~\ref{eq:a}, we obtain $z.scope.u < x.scope.l \leq y.scope.u$.


For the other direction,
Inequation \ref{eq:proposition2:eq0} yields that 
$x$ is a right relative of $z$ and it is in the scope of
the subtree of $y$, hence, $y$ is an ancestor of $x$ (($i$) of Proposition~\ref{prop:RPE2})
and $z$ is not an ancestor of $x$, so $y$ is the largest ancestor of $x$ that
belongs to $OT'$ (($ii$) of Proposition~\ref{prop:RPE2}).
Also, the preorder number of $x$
is larger than that of any vertex in the subtree of $z$ and hence of any
vertex in $OT'$ (($iii$) of Proposition~\ref{prop:RPE2}).
Hence, by Proposition~\ref{prop:RPE2}, $RExtend(OT',x,y)$ is an occurrence tree.
$\blacksquare$
\end{proof}

\begin{proposition}
\label{proposition:rightmostpathextensionb}
Let $OT'$ be an occurrence tree of a tree pattern $P'$ in a database tree $T$, $x \in V(T) \setminus  V(OT')$
and $y$ the rightmost vertex of $OT'$.
We have:
$RExtend(OT',x,y)$ is an occurrence tree iff
\begin{equation}
\label{eq:proposition3:eq0}
y.scope.l < x.scope.l  \text{ and }  x.scope.u \leq y.scope.u
\end{equation} 
\end{proposition}
\begin{proof}
First, assume $RExtend(OT',x,y)$ is an occurrence tree.
Similar to the proof of Proposition~\ref{proposition:rightmostpathextensiona}, 
by Proposition~\ref{prop:RPE1}, we get
\begin{equation}
\label{eq:a2}
y.scope.l < x.scope.l \leq x.scope.u \leq y.scope.u
\end{equation}

For the other direction, we assume $y.scope.l < x.scope.l$ and $x.scope.u \leq y.scope.u$.
This implies that the tree rooted
at $x$ is a subtree of the tree rooted at $y$ and that the preorder number
of $x$ is larger than the preorder number of $y$ and hence that all conditions of Proposition~\ref{prop:RPE2} are satisfied,
and $RExtend(OT',x,y)$ is an occurrence tree.
$\blacksquare$
\end{proof}

For example, in Figure \ref{fig:3}, 
first let $OT'$ refer to the occurrence tree consisting of a single vertex $0$.
The scopes of vertices $0$ and $1$ are $(0,5)$ and $(1,1)$, respectively.
The lower bound of the scope of vertex $1$ is greater than the 
the lower bound of the scope of vertex $0$;
and the upper bound of the scope of vertex $1$ is smaller than or equal to
the upper bound of the scope of vertex $0$.
Therefore, Inequation \ref{eq:proposition3:eq0} holds and $RExtend(OT',1,0)$ is an occurrence tree.
Now, let $OT'$ refer to the occurrence tree consisting of vertices $0$ and $1$. 
The scope of vertex $5$ is $(5,5)$.
The lower bound of the scope of vertex $5$ is greater than the upper bound of the scope of vertex $1$;
and it is smaller than or equal to the upper bound of the scope of vertex $0$.
Hence, Inequation \ref{eq:proposition2:eq0} holds and $RExtend(OT',5,0)$ is an occurrence tree.

\subsection{\textbf{Occ-list}: an efficient data structure for tree mining under subtree homeomorphism}
\label{sec:occ}

For the rightmost path extension of an occurrence tree, it suffices to know
its rightmost path. Different occurrences of a pattern can have the
same rightmost path. The key improvement over previous work is that we
only store information about the rightmost path of occurrence trees and
that different occurrence trees with the same rightmost path are
represented by the same data element. All occurrences of a pattern in
a database tree are represented by \textbf{occ-list}, a list of occurrences. An
element \textbf{occ} of \textbf{occ-list} represents all occurrences with a particular
rightmost path.
The element has four components:
\begin{itemize}
\item $TId$: the identifier of the database tree that contains the occurrences represented by \textbf{occ}. 
\item $scope$: the scope in the database tree $TId$ of the last vertex in the
rightmost path of the occurrences represented by \textbf{occ},
\item $RP$: an array containing the upper bounds of the scopes of the
vertices in the rightmost path of the occurrences represented by \textbf{occ},
i.e., with $x$ the vertex at rdepth $j$ in the rightmost path of the
pattern $P$ and $\varphi$ one of the occurrences represented by \textbf{occ},
$RP[j]= \varphi(x).scope.u$; note that this is the same value for all occurrences
$\varphi$ represented by \textbf{occ}\footnote{The upper bound of the scope of the
last vertex is already available in scope; for convenience of presentation, the information is duplicated in $RP$.}.
\item $ multiplicity$: the number of occurrences represented by \textbf{occ}.
\end{itemize}

For all occurrences of a pattern $P$ that have the same $TId$, $scope$ and $RP$,
one \textbf{occ} in the \textbf{occ-list} of $P$ is generated and 
its $ multiplicity $ shows the number of such occurrences.
All \textbf{occ}s of a pattern of size $1$ have $multiplicity$ $1$ as their
(single vertex) rightmost paths are all different.
Every occurrence is represented by exactly one \textbf{occ}.
We refer to the \textbf{occ-list} of $P$ by \textbf{occ-list}($P$). 
It is easy to see that the frequency of $P$ is equal to
$\sum_{oc \in \text{\textbf{occ-list}}(P)} oc.multiplicity$.
An example of \textbf{occ-list} is shown in Figure~\ref{fig:5}.

\begin{figure*}
\centering
\includegraphics[scale=0.7]{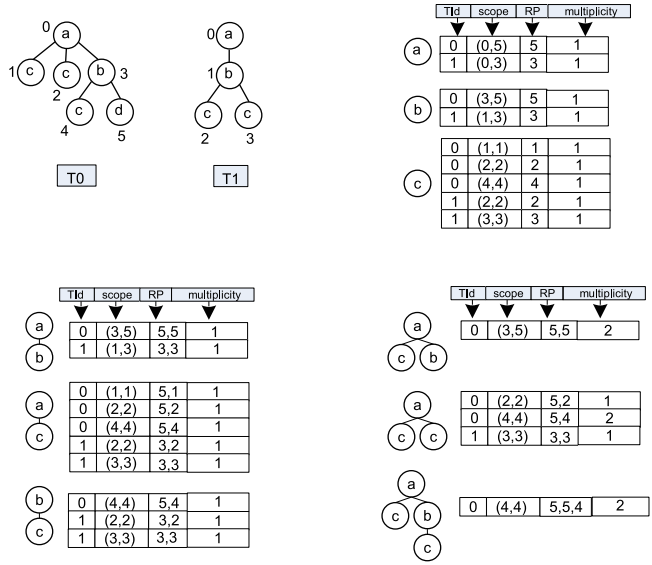}
\caption{
An example of \textbf{occ-list}. $T0$ and $T1$ are two database trees and minimum-support is equal to $2$.
The figure presents the \textbf{occ-list}s of some frequent 1-tree patterns,
frequent 2-tree patterns, frequent 3-tree patterns and frequent 4-tree patterns.
\label{fig:5}} 
\end{figure*}

\subsection{Operations on the \textbf{occ-list} data structure}
\label{sec:join}

Let $P$ be a tree of size $k+1$ generated by adding a vertex $v$ to a tree $P'$ of size $k$.
There are two cases:
\begin{enumerate}
\item $v$ is added to the rightmost vertex of $P'$. We refer to this case as \textit{leaf join}.
\item $v$ is added to a vertex in the rightmost path of $P'$, but not its rightmost vertex.
We refer to this case as \textit{inner join}.
\end{enumerate}
They correspond to Propositions~\ref{proposition:rightmostpathextensiona} and \ref{proposition:rightmostpathextensionb}.

\begin{proposition}[leaf join]
\label{proposition:leafjoin}
Let $v$ be a one element pattern, $P'$ a pattern with $u$ as the rightmost vertex,
$d = rdep_{P'}(u)$ and $P=RExtend(P',v,u)$.
Let \textbf{occ-list}($P'$), \textbf{occ-list}($v$)
and \textbf{occ-list}($P$) be the representations of the
occurrences of $P'$, $v$ and $P$, respectively.
We have: $oc \in  \textbf{occ-list}(P)$ iff there exists
an $oc' \in \textbf{occ-list}(P')$ and an $ov \in \textbf{occ-list}(v)$ such that
\begin{itemize}
\item[($i$)] 
$oc'.TId = ov.TId = oc.TId$ (all occurrences are from the same database tree),
\item[($ii$)]
$oc'.scope.l < ov.scope.l$ and $ov.scope.u \leq oc'.scope.u$,
\item[($iii$)]
$oc.RP[i] = oc'.RP[i]$ ($0 \leq i \leq d$) and $oc.RP[d+1]=ov.scope.u$ (i.e., copy of $oc'.RP$ and an extra element),
\item[($iv$)]
$oc.scope = ov.scope$, and
\item[($v$)]
$oc.multiplicity = oc'.multiplicity$
\end{itemize}
\end{proposition}
\begin{proof}
First, assume $oc \in  \textbf{occ-list}(P)$, hence it represents $oc.multiplicity$
occurrence trees of pattern $P$ in database tree $oc.TId$.
Each of these occurrence trees share the same rightmost path.
Hence they can be decomposed into occurrence trees of pattern $P'$ sharing the same rightmost path
and a particular occurrence of $v$. The latter is represented by an
element $ov$ of \textbf{occ-list}($v$).
The formers are represented by an element
$oc'$ of \textbf{occ-list}($P'$).
Now we show that conditions ($i$), \ldots, ($v$) hold between $oc$ and
these elements $oc'$ and $ov$.
Condition ($i$) holds because all occurrences are in database tree $oc.TId$,
($ii$) follows from Proposition~\ref{proposition:rightmostpathextensionb},
condition ($iii$) follows because the rightmost
path of the occurrence trees represented by $oc'$ is identical to that of $oc$,
except for the last element which is removed, ($iv$) follows from the
definition of $scope$, and ($v$) holds because $oc$ and $oc'$ represent the
same number of occurrences.

Second, assume there are an $oc' \in \textbf{occ-list}(P')$ and an $ov \in \textbf{occ-list}(v)$ such
that conditions ($i$), \ldots, ($v$) hold.
From ($i$) it follows that $oc'$ and $ov$ present
occurrences in the same database tree.
Because ($ii$) holds, if follows from Proposition ~\ref{proposition:rightmostpathextensionb} that all occurrence
trees of $P'$ represented by $oc'$ can be extended with $ov$ into occurrence
trees of $P$; all these occurrence trees have the same rightmost path
and have $ov$ as their rightmost vertex, moreover they are the only ones
in the database tree $oc'.TId$ with such a rightmost path.
Hence, the element $oc$ that satisfies properties ($i$), \ldots, ($v$) is indeed an
element of \textbf{occ-list}($P$) as has to be proven.
$\blacksquare$
\end{proof}

Figure~\ref{fig:6} illustrates the proposition. The proposition is the basis
for the $leaf\_join$ operation in Algorithm~\ref{algorithm:TPMiner} below.

\begin{figure*}
\centering
\includegraphics[scale=0.7]{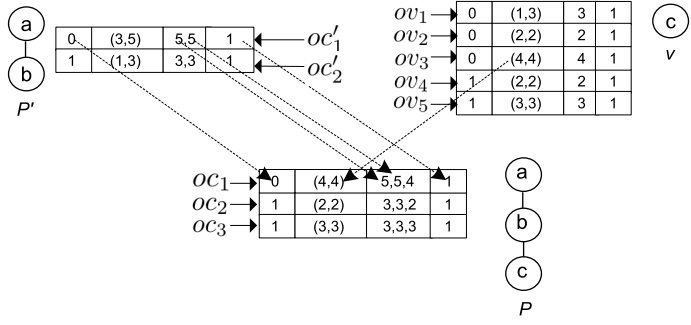}
\caption{\label{fig:6}
Details of the relationship between \textbf{occ-list}($P'$), \textbf{occ-list}($v$) and
\textbf{occ-list}($P$) for the database trees of Figure~\ref{fig:5}.
The entries $oc'_1$, $ov_3$ and $oc_1$ satisfy the properties of Proposition~\ref{proposition:leafjoin}.
Also the tuples $(oc'_2,ov_4, oc_2)$ and $(oc'_2, ov_5, oc_3)$ satisfy the properties.
The proposition is exploited in Algorithm \ref{algorithm:TPMiner} below.
Its $leaf\_join$ operation uses \textbf{occ-list}($P'$) and \textbf{occ-list}($v$) to compute \textbf{occ-list}($P$).
}
\end{figure*}

In contrast with leaf join, which is performed between \underline{\textit{one}} \textbf{occ} of a tree pattern and
\underline{\textit{one}} \textbf{occ} of a vertex,
inner join is performed between \underline{\textit{a set of \textbf{occ}s}} of a tree pattern and
\underline{\textit{one}} \textbf{occ} of a vertex.

\begin{proposition}[Inner join]
\label{proposition:innerjoin}
Let $v$ be a one element pattern, $P'$ a pattern,
$u$ a vertex on the rightmost path of $P'$ but not the rightmost vertex,
$c = rdep_{P'}(u)$ and $P=RExtend(P',v,u)$.
Let \textbf{occ-list}($P'$), \textbf{occ-list}($v$) and \textbf{occ-list}($P$)
be the representations of the occurrences of $P'$, $v$ and $P$, respectively.
We have: $oc \in \textbf{occ-list}(P)$ iff there exists 
a subset $oc'_1 , \hdots, oc'_m$ ($m \geq 1$) of \textbf{occ-list}($P'$)
and an $ov \in \textbf{occ-list}(v)$ such that
\begin{itemize}
\item[($i$)] 
$ov.TId=oc'_1.TId = \ldots = oc'_m.TId = oc.TId$ (all occurrences are from the same database tree),
\item[($ii$)]
$oc'_i.RP[k]= oc'_j.RP[k]$, for all $i,j \in [1..m]$ and for all $k \in [0..c]$,
\item[($iii$)]
$oc'_i.RP[c+1]  < ov.scope.l \leq oc'_i.RP[c]$ for all $i \in [1..m]$,
\item[($iv$)]
$oc'_1 , \hdots, oc'_m$ is maximal with the conditions ($i$)-($iii$),
\item[($v$)]
$oc.RP[i] = oc'.RP[i]$, $0 \leq i \leq c$, and $oc.RP[c+1]=ov.scope.u$ (copy of part of $oc'.RP$ and an extra element $ov.scope.u$),
\item[($vi$)]
$oc.scope = ov.scope$, and
\item[($vii$)]
$oc.multiplicity = \sum_{i=1}^m oc'_i.multiplicity $.
\end{itemize}
\end{proposition}

\begin{proof}
First, assume $oc \in \textbf{occ-list}(P)$, hence it represents $oc.multiplicity$
occurrence trees of pattern $P$ in database tree $oc.TId$.
All these occurrence trees share the same rightmost path.
Hence, they can be decomposed into occurrence trees of pattern $P'$
sharing the first $c+1$ vertices of the rightmost path
and a particular occurrence of $v$.
The latter is represented by an
element $ov$ of \textbf{occ-list}($v$).
The formers are represented by elements $oc'_1,\hdots, oc'_m$ of \textbf{occ-list}($P'$).
Now we show that conditions ($i$), \ldots, ($vii$) hold between $oc$ and
these elements $oc'_1,\hdots, oc'_m$ and $ov$.
Condition ($i$) holds because all occurrence trees are in database tree $oc.TId$,
($ii$) holds because all occurrence trees represented by $oc'_1,\hdots, oc'_m$ share the first $c+1$ vertices of the rightmost paths,
($iii$) and ($iv$) follow from Proposition~\ref{proposition:rightmostpathextensiona},
($v$) follows because the first $c+1$ vertices of the rightmost
paths of the occurrence trees represented by $oc$ are identical to those of $oc'_1,\hdots,oc'_m$,
and the rightmost vertex of the occurrence trees represented by $oc$ is the vertex represented by $ov$,
($vi$) follows from the definition of $scope$,
and ($vii$) holds because $oc$ represents $\sum_{i=1}^m oc'_i.multiplicity$ occurrences.

Second, assume there are a maximal subset $oc'_1, \hdots, oc'_m$ of $\textbf{occ-list}(P')$ and an $ov \in \textbf{occ-list}(v)$ such
that conditions (($i$), \ldots, ($vii$) hold.
From ($i$) it follows that $oc'$ and $ov$ present
occurrences in the same database tree.
Because ($iii$) and ($iv$) hold, if follows from Proposition~\ref{proposition:rightmostpathextensiona} that all occurrence
trees of $P'$ represented by $oc'_1,\hdots, oc'_m$ can be extended with $ov$ into occurrence
trees of $P$;
all these extended occurrence trees have the same rightmost path
and have $ov$ as their rightmost vertex, moreover they are the only ones
in the database tree $oc'.TId$ with such a rightmost path.
Hence, the element $oc$ that satisfies properties ($i$), \ldots, ($vii$) is indeed an
element of \textbf{occ-list}($P$) as has to be proven.
$\blacksquare$
\end{proof}

\begin{figure*}
\centering
\includegraphics[scale=0.77]{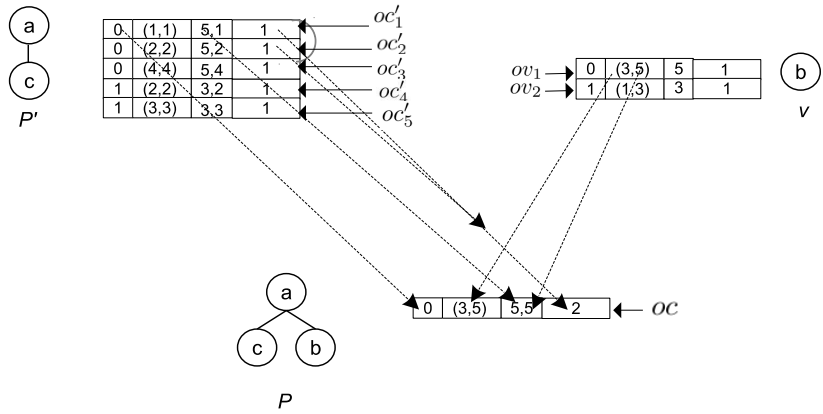}
\caption{\label{fig:7}
Details of the relationship between \textbf{occ-list}($P'$), \textbf{occ-list}($v$) and
\textbf{occ-list}($P$) for the database trees of Figure~\ref{fig:5}.
The entries $\{oc'_1,oc'_2\}$, $ov_1$ and $oc$ satisfy the properties of Proposition~\ref{proposition:innerjoin}.
As $c=0$, rightmost paths of the occurrence trees represented by $oc'_1$ and $oc'_2$ share the first vertex,
that is vertex $0$ of $T0$;
and rightmost paths of the occurrence trees represented by $oc$ share the first and second vertices, 
that are vertices $0$ and $3$ of $T0$.
The proposition is exploited in Algorithm~\ref{algorithm:TPMiner} below.
Its $inner\_join$ operation uses \textbf{occ-list}($P'$) and \textbf{occ-list}($v$) to compute \textbf{occ-list}($P$).
}
\end{figure*}

Figure~\ref{fig:7} illustrates the proposition. The proposition is the basis
for the $inner\_join$ operation in Algorithm~\ref{algorithm:TPMiner} below.



\subsection{Complexity analysis}
\label{sec:completely}


\paragraph{Space complexity.}

Given a database $\mathcal{D}$,
space complexity of the \textbf{occ-list}
of a pattern of size 1 is $O(n \times |\mathcal{D}|)$, where $n$ is the
maximum number of vertices that a database tree $T \in \mathcal{D}$ has.
For larger patterns, space complexity of the \textbf{occ-list} of a pattern
$P$ is $O( b \times d \times |\mathcal{D}|)$,
where
$b$ is the maximum number of occurrences with distinct rightmost paths that $P$ has in a database tree $T \in \mathcal{D}$, and
$d$ is the length of the rightmost path of $P$.

\begin{figure}
\begin{center}
\includegraphics[width=0.45\textwidth]{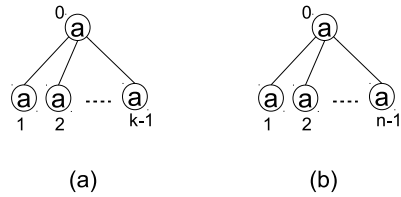}
\end{center}
\caption{\label{fig:complexity1}
Figure (a) shows a tree pattern $P$ and (b) a database tree $T$.
The number of occurrences of $P$ in $T$ is exponential in terms of $n$ and $k$.
In this case, the size of \textbf{occ-list} is linear, however, the size of the data-structures generated by the other algorithms
is exponential.
}
\end{figure}

Compared to data-structures generated by other algorithms such as VTreeMiner,
\textbf{occ-list} is often substantially more compact.
Given a database $\mathcal{D}$, the size of the data-structure generated by VTreeMiner for a pattern $P$ is $O( e \times k \times |\mathcal{D}|)$,
where $e$ is the maximum number of occurrences that $P$ has in a database tree $T \in \mathcal{D}$ and $k$ is $|V(P)|$.
We note that on one hand, $k \geq d $ and the other hand, $e \geq b$.
In particular, $e$ can be significantly larger than $b$, e.g., it can be exponentially (in terms of $n$ and $k$) larger than $b$. 
Therefore, in the worst case, the size of the data-structure generated by VTreeMiner is exponentially larger than \textbf{occ-list}
(and it is never smaller than \textbf{occ-list}).
An example of this situation is shown in Figure~\ref{fig:complexity1}.
In this figure, pattern $P$ has 
$
{n-1 \choose k-1} \geq  { \left( \frac{n-1}{k-1} \right) }^{k-1} \geq { \left( \frac{n}{k} \right) }^{k / 2} 
$
occurrences in the database tree $T$ ($k > 2$ and $n \geq k$).
If $k = (n+1) / 2 $, the size of the data-structure generated by VTreeMiner will be $\Omega(2^{n/4}) $.
However, in this case, the size of the \textbf{occ-list} generated for $P$ is $\Omega(n) $.
More precisely, \textbf{occ-list}($P$) has $n-k+1$ \textbf{occ}s,
where for each $i$, $ k-1 \leq i \leq n-1 $, there exists an \textbf{occ} with
$TId=0$, $scope=(i,i)$, $RP=\{n-1,i\}$ and $multiplicity={i-1 \choose k-2}$.

\begin{figure}
\begin{center}
\includegraphics[width=0.31\textwidth]{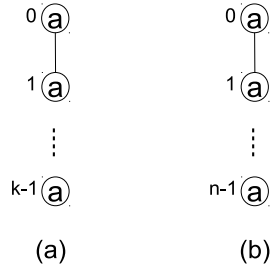}
\end{center}
\caption{\label{fig:complexity2}
Figure (a) shows a tree pattern $P$ and (b) a database tree $T$.
The number of occurrences of $P$ in $T$ is exponential in terms of $n$ and $k$.
In this case, the size of the \textbf{occ-list} and the size of the data-structures generated by the other algorithms for $P$ are exponential.
}
\end{figure}

We note that there are cases where the size of \textbf{occ-list} 
(as well as the size of the data-structures used by the other algorithms) becomes exponential.
An example of this situation is presented in Figure~\ref{fig:complexity2},
where for $k = n / 2 $, the size of \textbf{occ-list} as well as the size of \textit{scope-list}
used by VTreeMiner become exponential.
More precisely,
for each $k$-element combination of the vertices of the database tree,
there exists an \textbf{occ} in the \textbf{occ-list} of the tree pattern, where
$TId$ is $0$, the lower bound of $scope$ is the preorder number of the vertex with the largest depth in the combination,
the upper bound of $scope$ is $n-1$,
$RP$ is an array of size $k$ filled by $n-1$,
and $multiplicity$ is $1$.

In most of real-world databases, such as CSLOGS \cite{36} and NASA \cite{Chalmers2},
while several vertices of a database tree have the same label,
no two vertices on the same path are identically labeled.
For trees with this property, while worst case space complexity of \textbf{occ-list} is linear,
worst case size of \textit{scope-list} remains exponential.




\paragraph{Time complexity.}
We study time complexity of leaf join and inner join:
\begin{itemize}
\item 
A leaf join between two \textbf{occ}s takes $O(d)$ time
with $d$ the length of the rightmost path of $P$.
Since a pattern larger than 1 has $O( b \times |\mathcal{D}|)$ \textbf{occ}s 
and a pattern of size 1 has $O( n \times |\mathcal{D}|)$ \textbf{occ}s 
and leaf join is performed between every pairs of \textbf{occ}s with the same $TId$,
worst case time complexity of leaf join between two \textbf{occ-list}s will be $O(d \times b \times n \times |\mathcal{D}|)$.
\item
In the inner join of the \textbf{occ-list}s of a tree pattern $P'$ and a vertex $v$, given an \textbf{occ} $ov$ of $v$,
it takes $O(h \times d)$ time 
to find subsets of the \textbf{occ-list} of $P'$ that satisfy the conditions of Proposition~\ref{proposition:innerjoin},
where $h$ is the number of \textbf{occ}s in the \textbf{occ-list} of $P'$.
We note that \textbf{occ}s of an \textbf{occ-list} can be automatically sorted
first based on their $TId$s and second, based on their $RP$s.
This makes it possible to find the subsets of the \textbf{occ-list} of $P'$
that satisfy the conditions of Proposition~\ref{proposition:innerjoin}
only by one scan of the \textbf{occ-list} of $P'$.
During this scan, it is checked whether:  
(i) the current and the previous \textbf{occ}s have the same $TId$,
(ii) the current and the previous \textbf{occ}s have the same $RP[0],\hdots,RP[c]$,
(iii) $RP[c+1]$ of the current \textbf{occ} is less than $ov.scope.l$, and
(iv) $RP[c]$ of the current \textbf{occ} is greater than or equal to $ov.scope.l$.
During the scan of \textbf{occ-list}($P'$), after finding a subset $S$
that satisfies the conditions of Proposition~\ref{proposition:innerjoin},
it takes $O(d)$ time to perform the inner join of $S$ and $ov$.
As a result, it takes $O(h \times d)$ time to perform the inner join of the \textbf{occ-list} of $P'$ and
an \textbf{occ} of $v$.
Since $h$ is $O( b \times |\mathcal{D}|)$ and
$v$ has $O(n \times|\mathcal{D}|)$ \textbf{occ}s,
and since inner join is done between every pairs of subsets $S$ and \textbf{occ}s $ov$ that have the same $TId$s,
time complexity of inner join will be $O( d \times b \times n \times |\mathcal{D}|)$.
\end{itemize}
Therefore, frequency of a pattern can be counted in $O( d \times b \times n \times |\mathcal{D}|)$ time.
We note that in VTreeMiner, frequency of a pattern is counted in $O( k \times e^2 \times |\mathcal{D}|)$ time
(recall $k \geq d$ and $e\geq b$, in particular, it is possible that $e >> b$).

\subsection{A brief comparison with other vertical frequency counting approaches}
\label{sec:briefcomparison}

As mentioned before, algorithms such as VTreeMiner \cite{36} and MB3Miner \cite{23}
need to keep track of all occurrences.
Obviously, our approach is more efficient as it simultaneously
represents and processes all occurrences sharing the rightmost path in $O(d)$ space,
where $d$ is the length of the rightmost path of the pattern.
TreeMinerD \cite{36}, developed by Zaki for computing \textit{per-tree frequency},
is more similar to our approach as it also processes rightmost paths of occurrences. 
However, there are significant differences.
First, while TreeMinerD computes only \textit{per-tree frequency},
our algorithm performs a much more expensive task and computes \textit{per-occurrence frequency}.
Second, TreeMinerD applies different join strategies.

\section{TPMiner: an efficient algorithm for finding frequent embedded tree patterns}
\label{sec:tpminer}

In this section, we first introduce the TPMiner algorithm and then,
we discuss an optimization technique used to reduce the number of generated infrequent patterns.

\subsection{The TPMiner algorithm}

Having defined the operations of leaf join and inner join and having
analyzed their properties, we can now introduce our tree pattern miner,
called TPMiner (\textbf{T}ree \textbf{P}attern \textbf{Miner}).
TPMiner builds all frequent patterns and maintains the rightmost paths of
all their occurrences.

\begin{algorithm}
\caption{High level pseudo code of the TPMiner algorithm.\label{algorithm:TPMiner}}
\begin{algorithmic} [1]
\STATE \textsf{TPMiner}
\STATE \textbf{Input:} $\mathcal{D}$ \COMMENT{a set of  database trees}, $ minsup$ \COMMENT{the  minimum support threshold}
\STATE \textbf{Output:} $\mathcal{P}$ \COMMENT{the set of frequent patterns}
\STATE Compute the set $\mathcal{P}_1$ of frequent patterns of size $1$ along with their \textbf{occ-list}s
\STATE $\mathcal{P} \leftarrow \mathcal{P}_1$
\FORALL{$P$ in $\mathcal{P}_1$}
\STATE \textsf{Extend}($P$,$\mathcal{P}_1$, $minsup$, $\mathcal{P}$)
\ENDFOR
\\\hrulefill
\end{algorithmic}
\begin{algorithmic} [1]
\STATE \textsf{Extend}($P$,$\mathcal{P}_1$, $minsup$, $\mathcal{P}$)
\STATE \textbf{Input:} $P$ \COMMENT{a frequent pattern}, $\mathcal{P}_1$ \COMMENT{the set of frequent patterns of size 1},
$minsup$ \COMMENT{the  minimum support threshold}
\STATE \textbf{Input and Output:} $\mathcal{P}$ \COMMENT{the set of frequent patterns}
\STATE \textbf{Side effect:} $\mathcal{P}$ is updated with frequent rightmost path extensions of $P$
\FORALL{$P_1$ in $\mathcal{P}_1$}
\STATE \COMMENT{Let $d$ be the length of the rightmost path of $P$}
\FOR{$i=0$ \textbf{to} $d$}
\STATE \COMMENT{Let $u$ be the vertex of $P$ such that $rdep_{P}(u)=i$}
\STATE $P_n \leftarrow RExtend(P, P_1, u)$
\IF{$i=d$}
\STATE \textbf{occ-list}$(P_n)\leftarrow leaf\_join(P,P_1)$
\ELSE
\STATE \textbf{occ-list}$(P_n) \leftarrow inner\_join(P,P_1)$
\ENDIF
\IF{$sup(P_n,\mathcal D) \geq minsup$}
\STATE Add $P_n$ to $\mathcal{P}$
\STATE \textsf{Extend}($P_n$,$\mathcal{P}_1$, $minsup$, $\mathcal{P}$)
\ENDIF
\ENDFOR
\ENDFOR
\end{algorithmic}
\end{algorithm}

TPMiner follows a depth-first strategy for generating tree patterns.
First, it extracts frequent patterns of size $1$ (frequent vertex labels) and constructs their \textbf{occ-list}s.
This step can be done by one scan of the database.
Every \textbf{occ} of a pattern of size $1$ represents one occurrence of the pattern,
where its $RP$ contains the upper bound of the scope of the occurrence,
its $scope$ contains the scope of the occurrence, and its $multiplicity$ is $1$.
Then, larger patterns are generated using rightmost path extension.
For every tree $P$ of size $k+1$ ($k\geq 1$) which is generated by adding a vertex $v$ to a vertex on the rightmost path of a tree $P'$,
the algorithm computes the \textbf{occ-list} of $P$ by joining the \textbf{occ-list}s of $P'$ and $v$.
If $v$ is added to the the rightmost vertex of $P'$, a \textit{leaf join} is performed;
otherwise, an \textit{inner join} is done.
The high level pseudo code of TPMiner is given in Algorithm~\ref{algorithm:TPMiner}.
$\mathcal P$ is used to store all frequent patterns.
Every tree pattern is generated in $O(1)$ time, hence, time complexity of Algorithm~\ref{algorithm:TPMiner}
is $O( d \times b \times n \times |\mathcal{D}| \times \mathcal C)$,
where $\mathcal C$ is the number of generated candidates.


\subsection{An optimization technique for candidate generation}

In rightmost path extension, a new tree $P$ is generated by adding a new vertex to some vertex on the rightmost path of an existing frequent pattern $P'$, therefore,
it is already known that $P$ has (at least) one frequent subtree.
In an improved method proposed in \cite{36}, to generate the tree $P$, two patterns $P1$ and $P2$ such that their subtrees induced by all but the rightmost vertices
are the same, are merged. In this merge, the rightmost vertex of $P2$ (which is not in $P1$) is added to $P1$ as the new vertex.
In this way, we already know that $P$ has (at least) two subtrees that are frequent, therefore, trees that have only one frequent subtree are not generated.
This can reduce the number of trees that are generated but are infrequent.

\section{Experimental Results}
\label{sec:experimentalresults}

We performed extensive experiments to evaluate the efficiency of the proposed algorithm,
using data from real applications as well as synthetic datasets.
The experiments were done on one core of a single AMD
Processor $270$ clocked at $2.0$ GHz
with $16$ GB main memory and $2 \times 1$ MB L2 cache,
running Ubuntu Linux $12.0$.
The program was compiled by the GNU C++ compiler $4.0.2$.

VTreeMiner (also called TreeMiner) \cite{36} is a well-known algorithm
for finding all frequent embedded patterns from trees.
Therefore, we select this algorithm for our comparisons.
Recently more efficient algorithms, such as TreeMinerD \cite{36}, XSpanner \cite{26}, TRIPS and TIDES \cite{24},
have been proposed for finding frequent embedded tree patterns.
However, they only compute the per-tree frequency instead of the per-occurrence frequency.
Some other algorithms, such as \cite{30} and \cite{17}, find maximal embedded patterns
which are a small subset of all frequent tree patterns.
To the best of our knowledge at the time of writing this paper,
MB3Miner \cite{23} is the most efficient recent algorithm for finding frequent embedded tree patterns.
MB3Miner works with both per-tree frequency and per-occurrence frequency. 
Here, we use the version that works with the per-occurrence frequency.
We note MB3Miner generates only the frequent patterns such that all subtrees are frequent,
therefore, it might produce fewer frequent patterns than VTreeMiner and TPMiner.
Tatikonda et al. \cite{Tatikonda2} proposed efficient techniques for parallel mining of trees on multicore systems.
Since we do not aim at parallel tree mining, their system is not proper for our comparisons.

\begin{table}
\caption{\label{table:datase} Summary of real-world datasets.}
\begin{center}
\begin{tabular}{ l | l l l l l}
\hline
\multirow{2}{*}{Dataset}       & \multirow{2}{*}{\# Transactions} & \multirow{2}{*}{\# Vertices}  & \multirow{2}{*}{\# Vertex labels}  & \multicolumn{2}{c}{Transaction size} \\
\cline{5-6}
                               &                                  &                               &                                    &   Maximum   &  Average    \\ \hline
CSLOGS32241  & $32,241$        & $240,716$   &  $10,698 $  &  $435 $  &  $13.9323 $ \\
Prions       & $ 17,551 $      & $227,203 $ & $111 $ & $37 $  &  $24.9497$ \\
NASA         & $1,000 $        &  $163,753 $ &  $333 $ &  $471 $ & $326.506$\\
\hline
\end{tabular}
\end{center}
\end{table}

We used several real-world datasets from different areas to evaluate the efficiency of TPMiner.
The datasets do neither have noise nor missing values.
Table \ref{table:datase} reports basic statistics of our real-world datasets.

The first real-world dataset is CSLOGS \cite{36}
that contains the web access trees of the CS department of the Rensselaer
Polytechnic Institute.
It has $59,691$ trees, $716,263$ vertices and $13,209$ vertex labels.
Each distinct label corresponds to the URLs of a web page.
As discussed in \cite{23}, when \textit{per-occurrence} frequency is used,
none of the algorithms can find meaningful patterns,
because by decreasing $minsup$, suddenly lots of frequent patterns with many occurrences appear in the dataset
that makes the mining task practically intractable.
The authors of \cite{23} progressively reduced the dataset and
generated the CSLOGS32241 dataset that contains $32,241$ trees.
Figure~\ref{figure:CSLOGS32241} compares the algorithms over CSLOGS32241.
For all given values of $minsup$, VTreeMiner does not terminate within a reasonable time (i.e., $1$ day!),
therefore, the figure does not contain it.
TPMiner is faster than MB3Miner by a factor of $5$-$20$ and it requires significantly less memory cells.
In order to have a comparison with VTreeMiner, we tested the algorithms for $minsup=1000$;
while TPMiner finds frequent patterns within around $1.3$ seconds, 
VTreeMiner takes more than $1000$ seconds to find the same patterns.

\begin{figure*}
\centering
\subfigure[The horizontal axis shows the minimum support; the left vertical axis
the running time (sec) and the right vertical axis the number of patterns generated by TPMiner (and VTreeMiner).]
{
\includegraphics[scale=0.35]{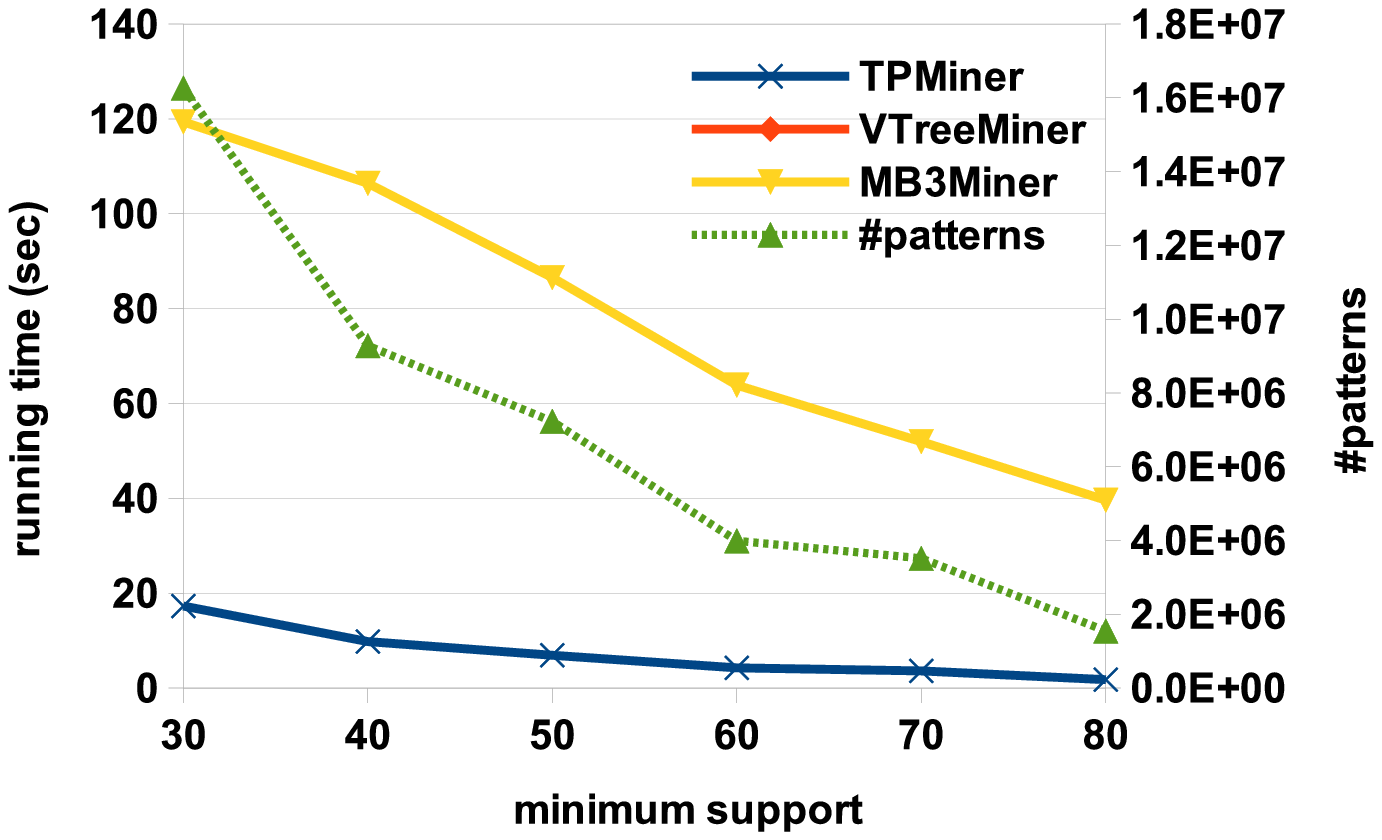}
\label{fig:CSLOGS32241_time}
}
\quad
\subfigure[The horizontal axis shows the minimum support; the left vertical axis
the memory usage (Byte) and the right vertical axis the number of patterns. The vertical axes are in logarithmic scale.]
{
\includegraphics[scale=0.35]{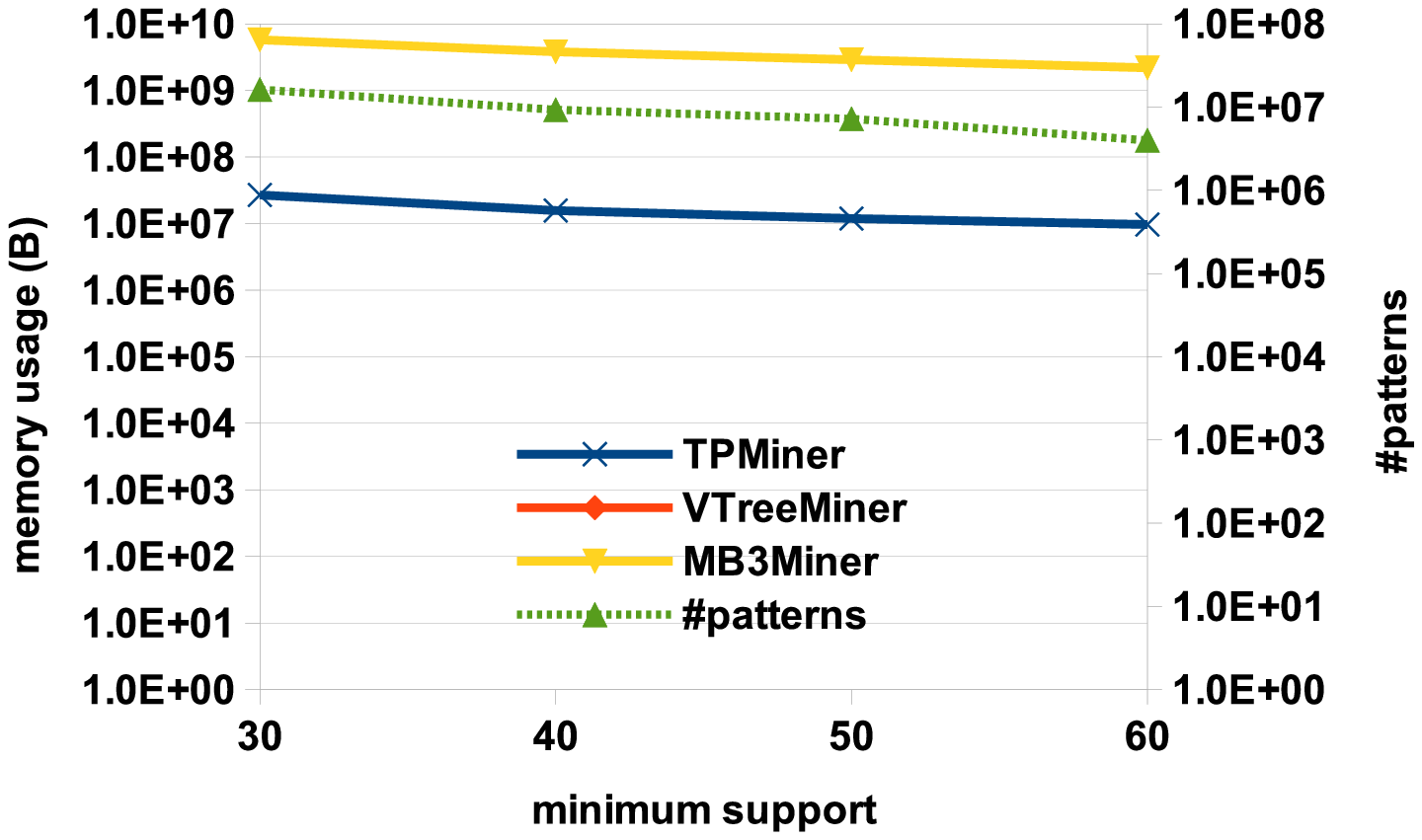}
\label{fig:CSLOGS32241_mem}
}
\caption
{
\label{figure:CSLOGS32241}
Comparison over CSLOGS32241.
}
\end{figure*}

\begin{figure*}
\centering
\subfigure[Minimum support vs. running time.]
{
\includegraphics[scale=0.35]{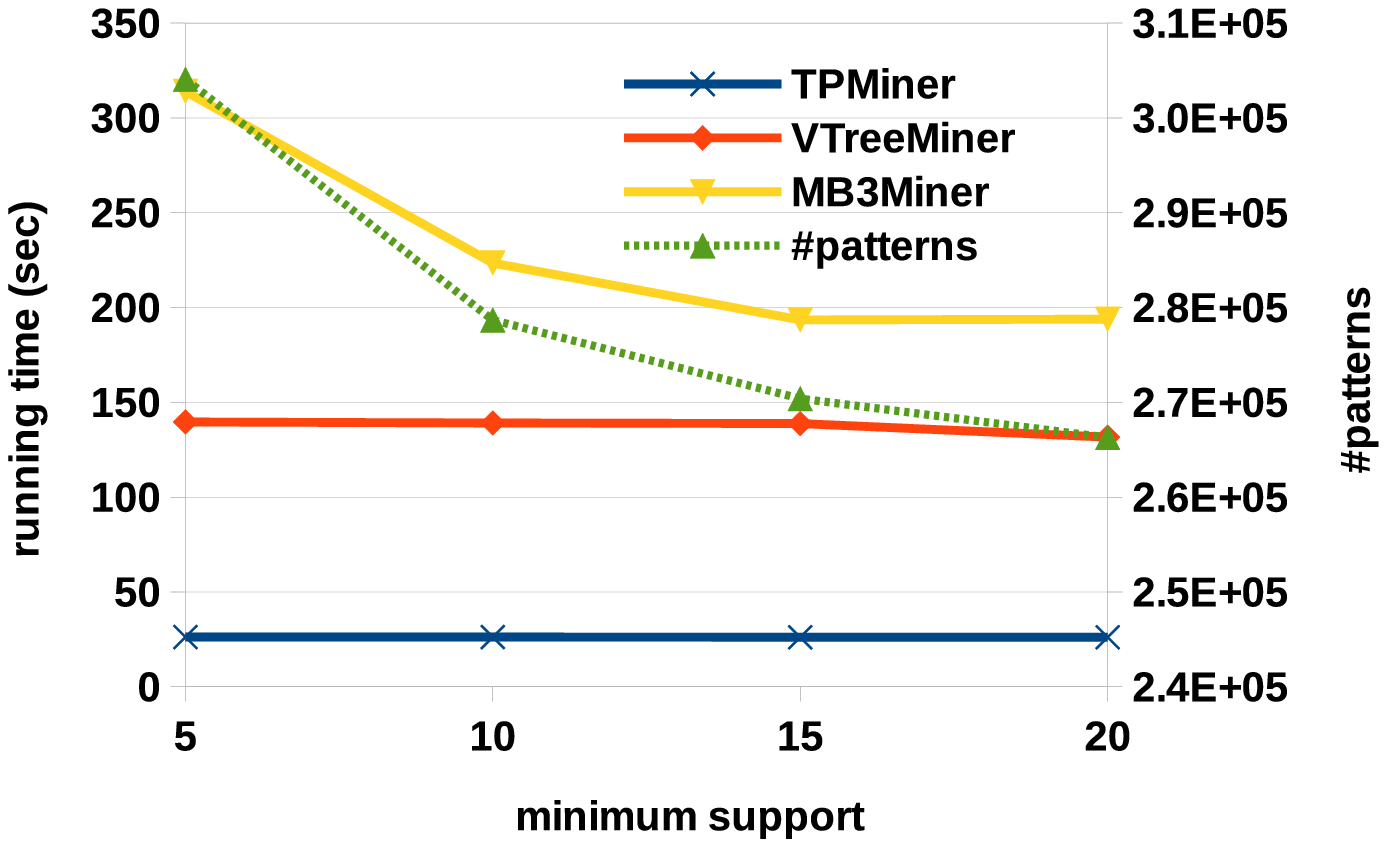}
\label{fig:prions_time}
}
\subfigure[Minimum support vs. memory usage. The vertical axes are in logarithmic scale.]
{
\includegraphics[scale=0.35]{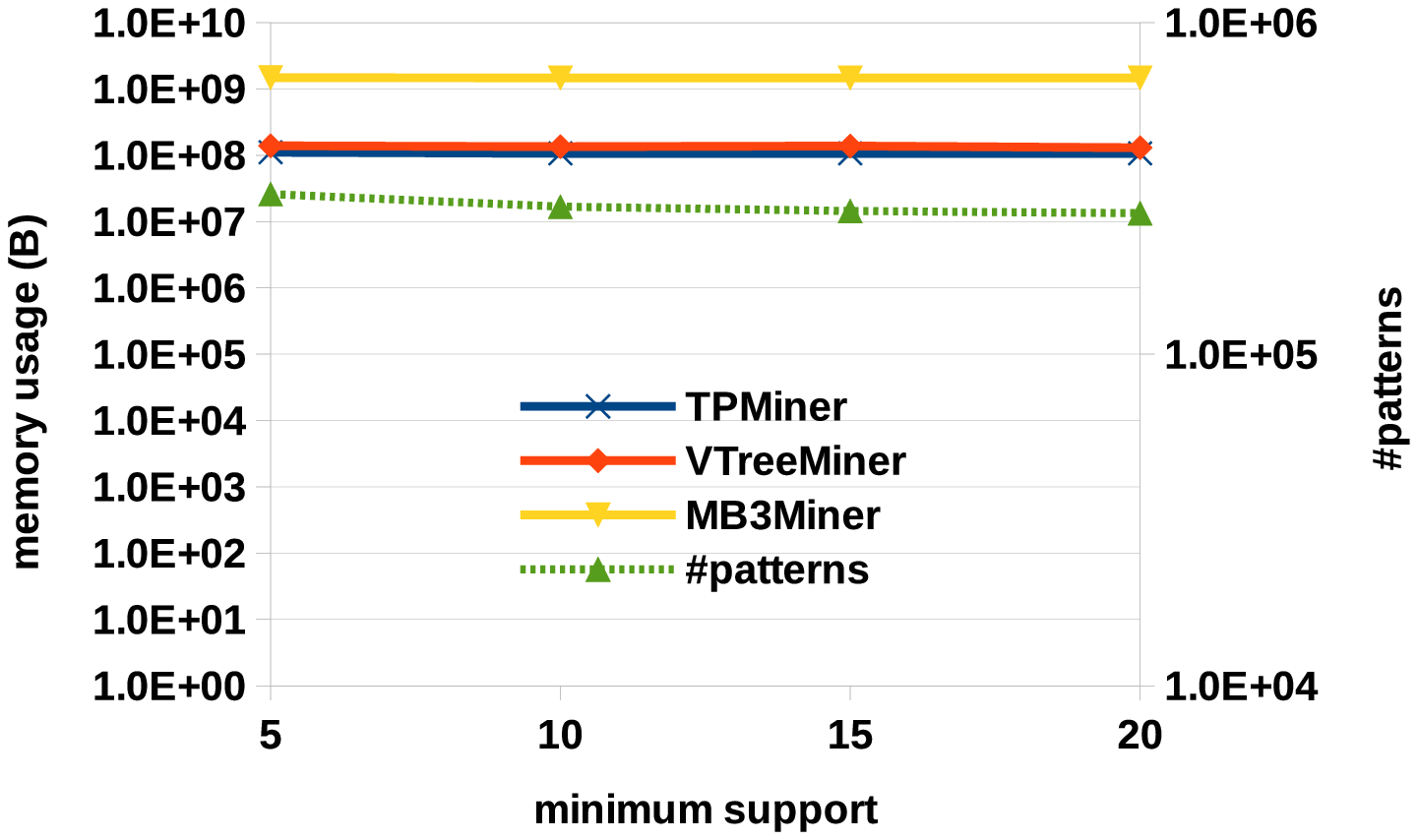}
\label{fig:prions_mem}
}
\caption
{
\label{figure:prions}
Comparison over Prions.
}
\end{figure*}

The second real-world dataset used in this paper is Prions that
describes a protein ontology database for Human Prion proteins in XML format \cite{Sidhu}.
The authors of \cite{23} converted it into a tree-structured dataset by considering tags as vertex labels.
It has $17,551$ wide trees.
Figure~\ref{figure:prions} reports the empirical results over this dataset,
where TPMiner is faster than VTreeMiner by a factor of $5$-$5.2$,
and it is faster than MB3Miner by a factor of $7.3$-$11$.
The third real-world dataset is a dataset of IP multicast.
The NASA dataset consists of MBONE multicast data that was measured during the NASA shuttle launch between the $14$th
and $21$st of February, $1999$ \cite{Chalmers1} and \cite{Chalmers2}.
It has $333$ distinct vertex labels where each vertex label is an IP address.
The data was sampled from this NASA dataset with $10$ minutes sampling interval and has $1,000$ trees.
In this dataset, large frequent patterns are found at high minimum support values.
As depicted in Figure~\ref{figure:nasa}, over this dataset,
TPMiner is $3$-$4$ times faster than VTreeMiner and both methods are significantly faster than MB3Miner.
At $minsup=902$, MB3Miner fails.

\begin{figure*}
\centering
\subfigure[Minimum support vs. running time.]
{
\includegraphics[scale=0.35]{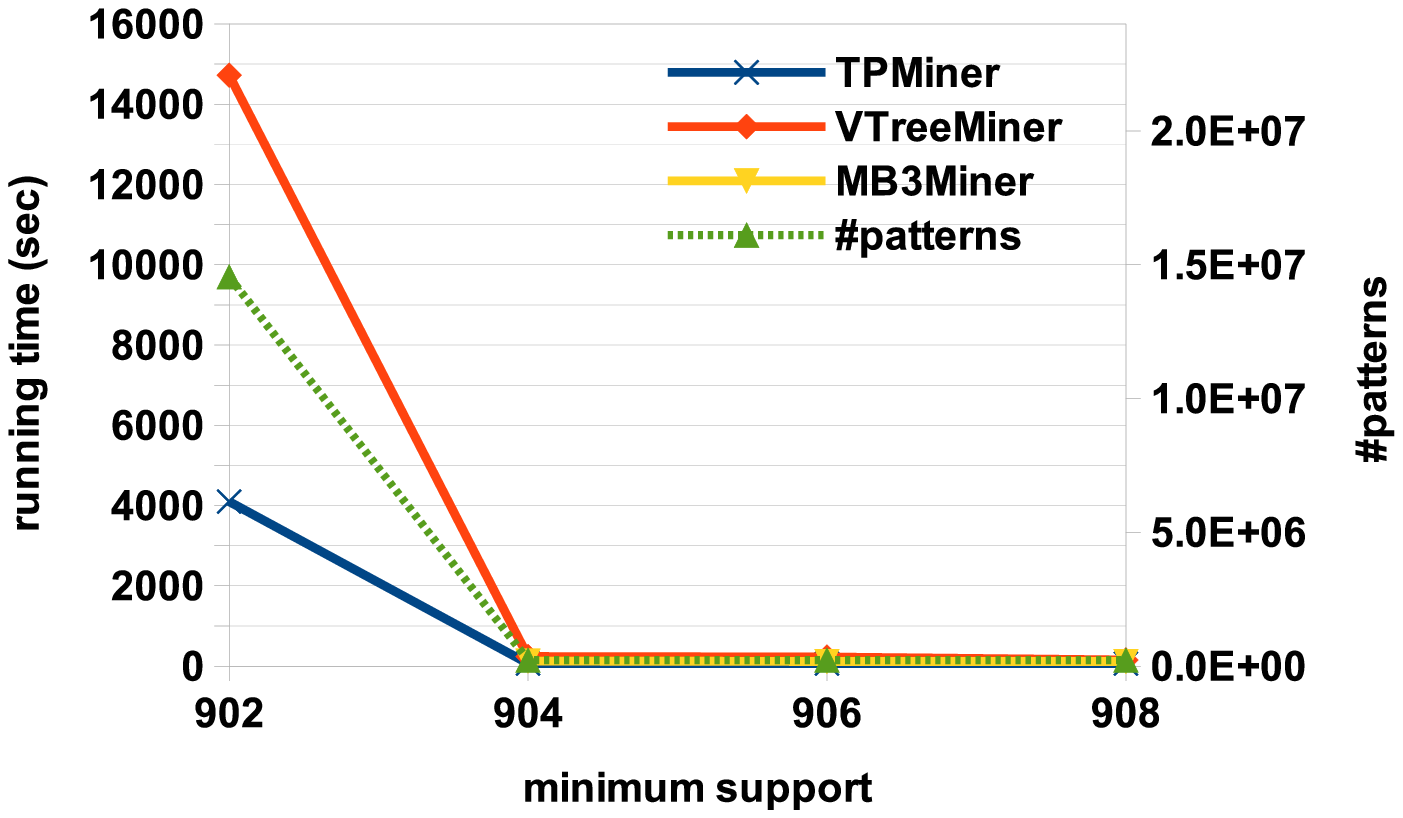}
\label{fig:nasa_time}
}
\subfigure[Minimum support vs. memory usage. The vertical axes are in logarithmic scale.]
{
\includegraphics[scale=0.35]{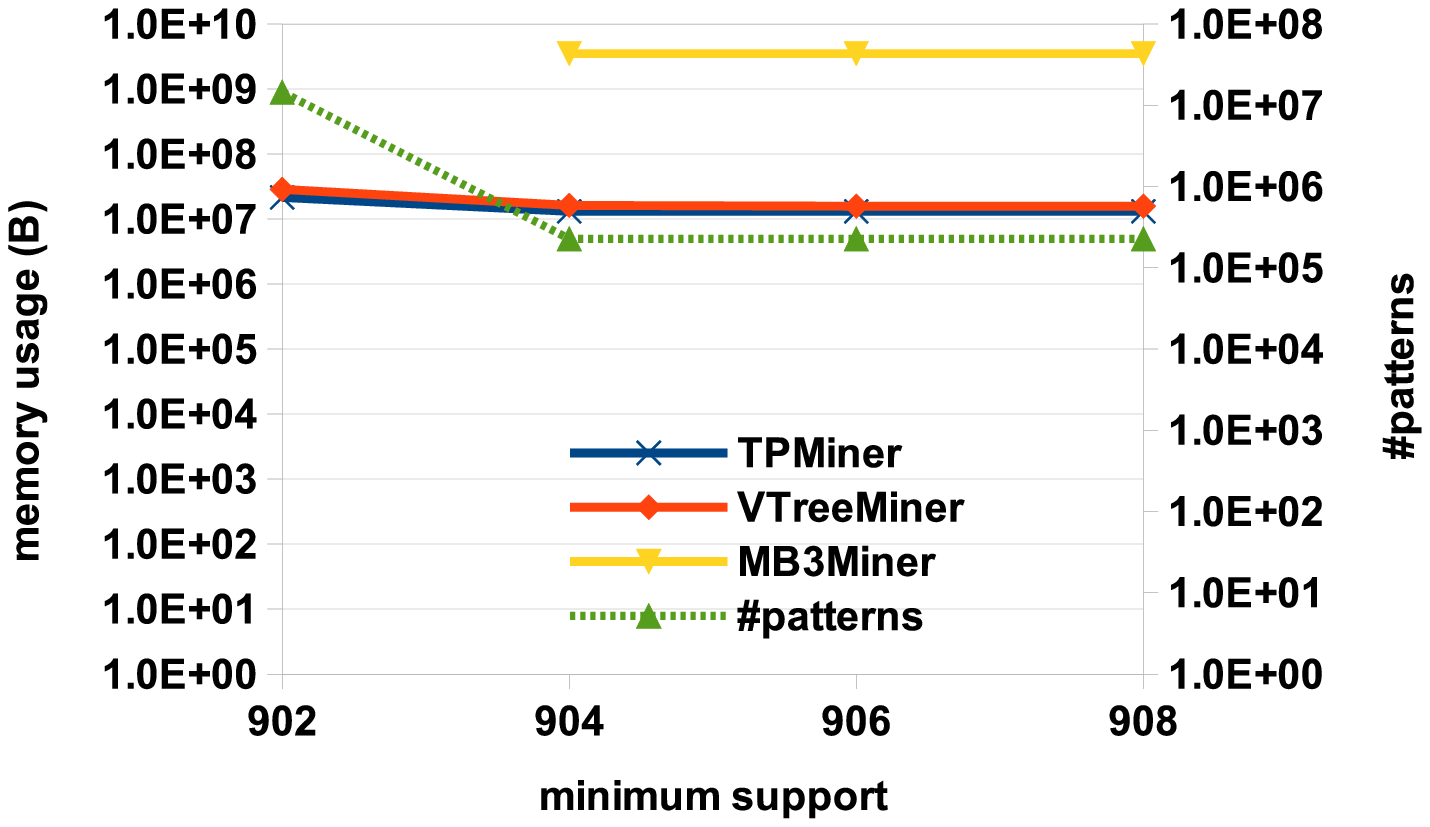}
\label{fig:nasa_mem}
}
\caption
{
\label{figure:nasa}
Comparison over NASA.
}
\end{figure*}

\begin{figure*}
\centering
\subfigure[D10: minimum support vs. running time.]
{
\includegraphics[scale=0.35]{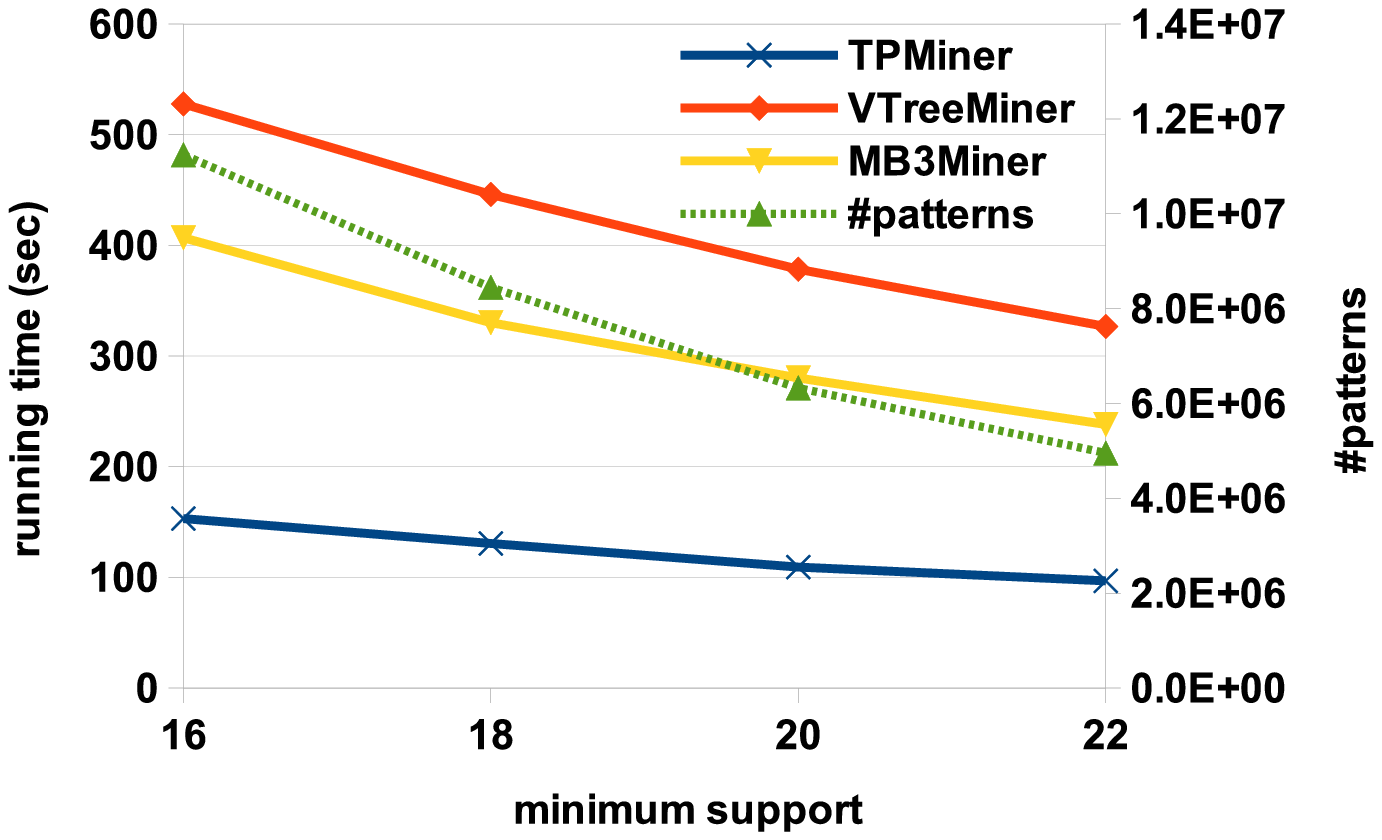}
\label{fig:d10_time}
}
\subfigure[D10: minimum support vs. memory usage. The vertical axes are in logarithmic scale.]
{
\includegraphics[scale=0.35]{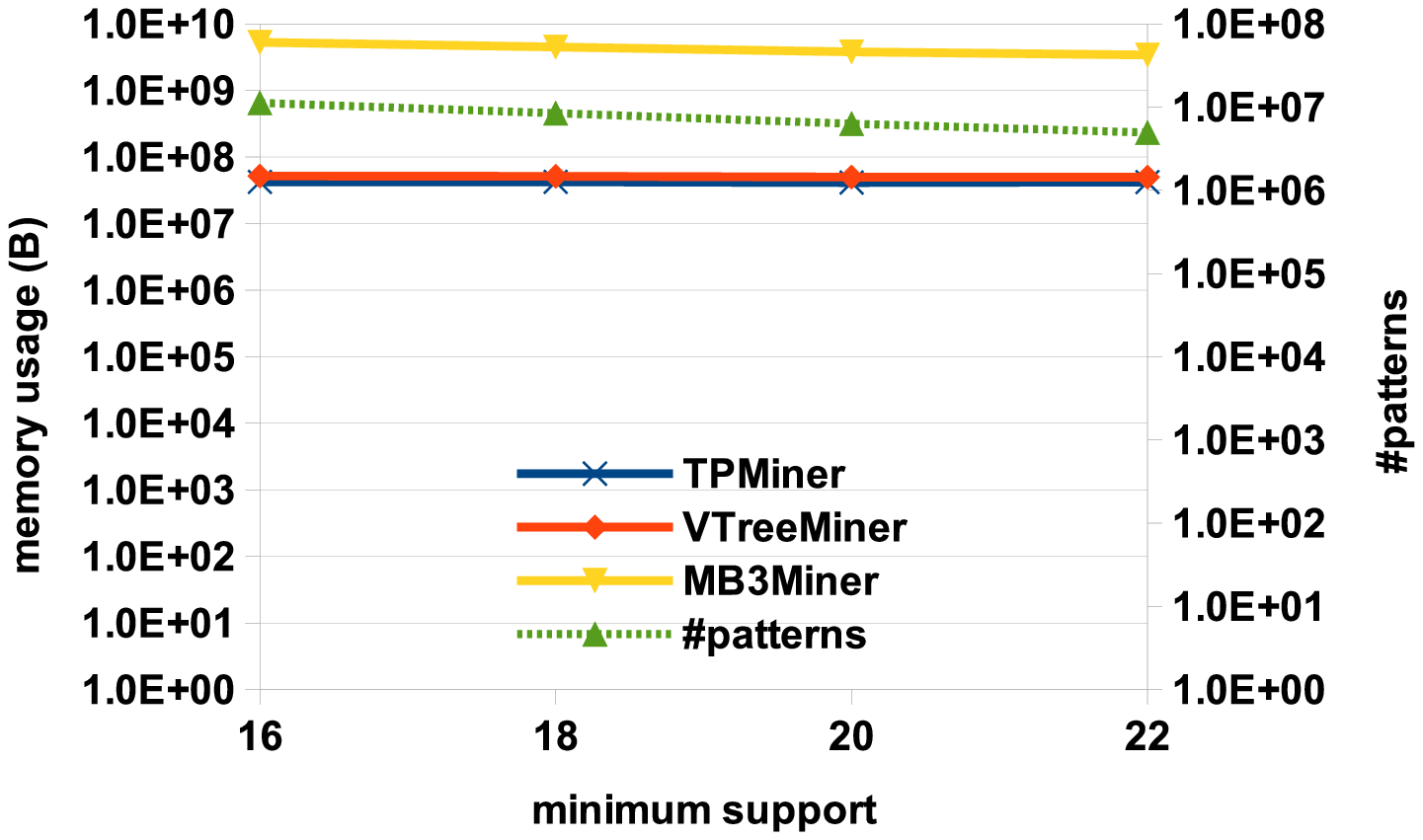}
\label{fig:d10_mem}
}
\subfigure[D5: minimum support vs. running time.]
{
\includegraphics[scale=0.35]{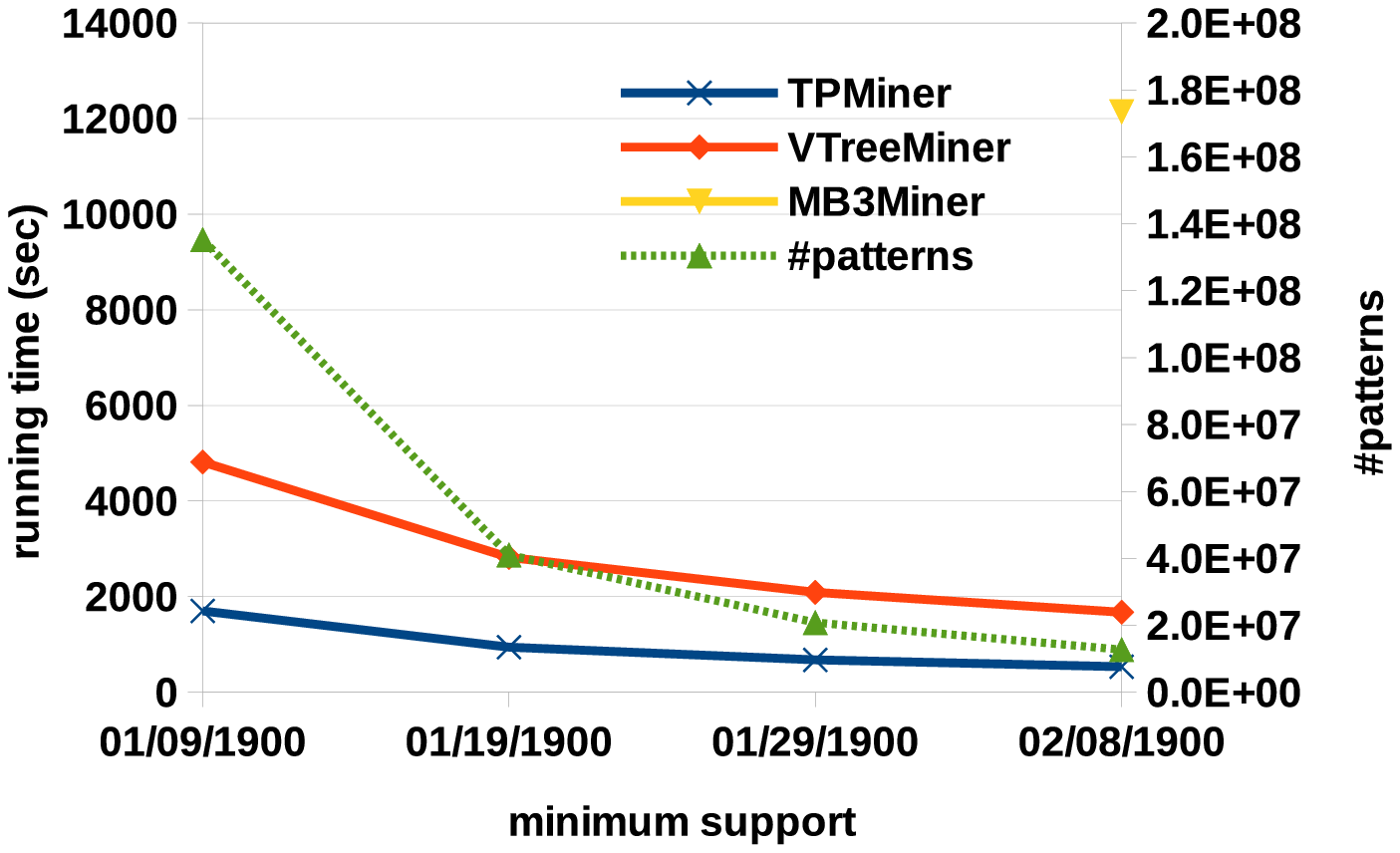}
\label{fig:d5_time}
}
\subfigure[D5: minimum support vs. memory usage. The vertical axes are in logarithmic scale.]
{
\includegraphics[scale=0.35]{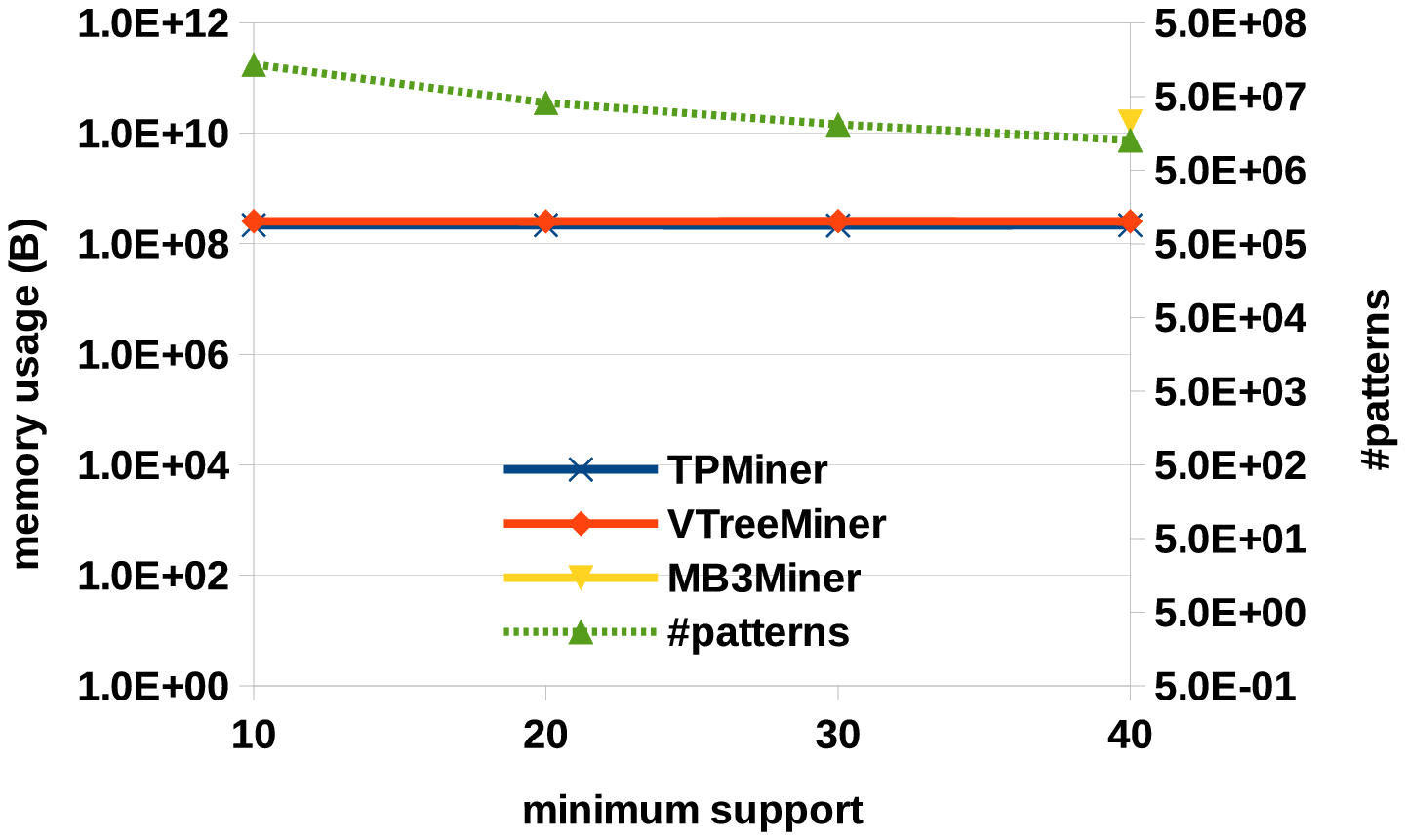}
\label{fig:d5_mem}
}
\subfigure[M10K: minimum support vs. running time.]
{
\includegraphics[scale=0.35]{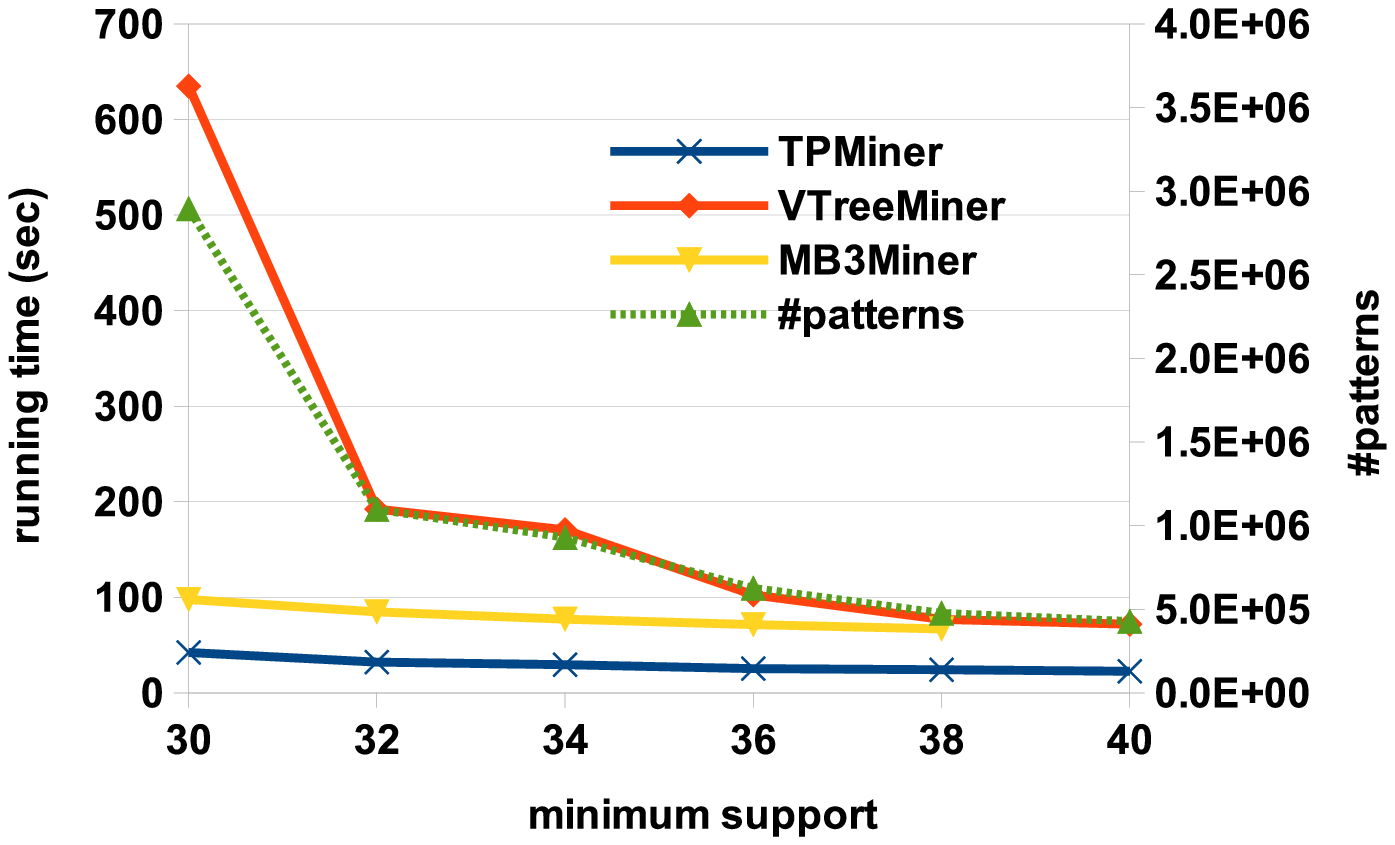}
\label{fig:m100000_time}
}
\subfigure[M10K: minimum support vs. memory usage. The vertical axes are in logarithmic scale.]
{
\includegraphics[scale=0.35]{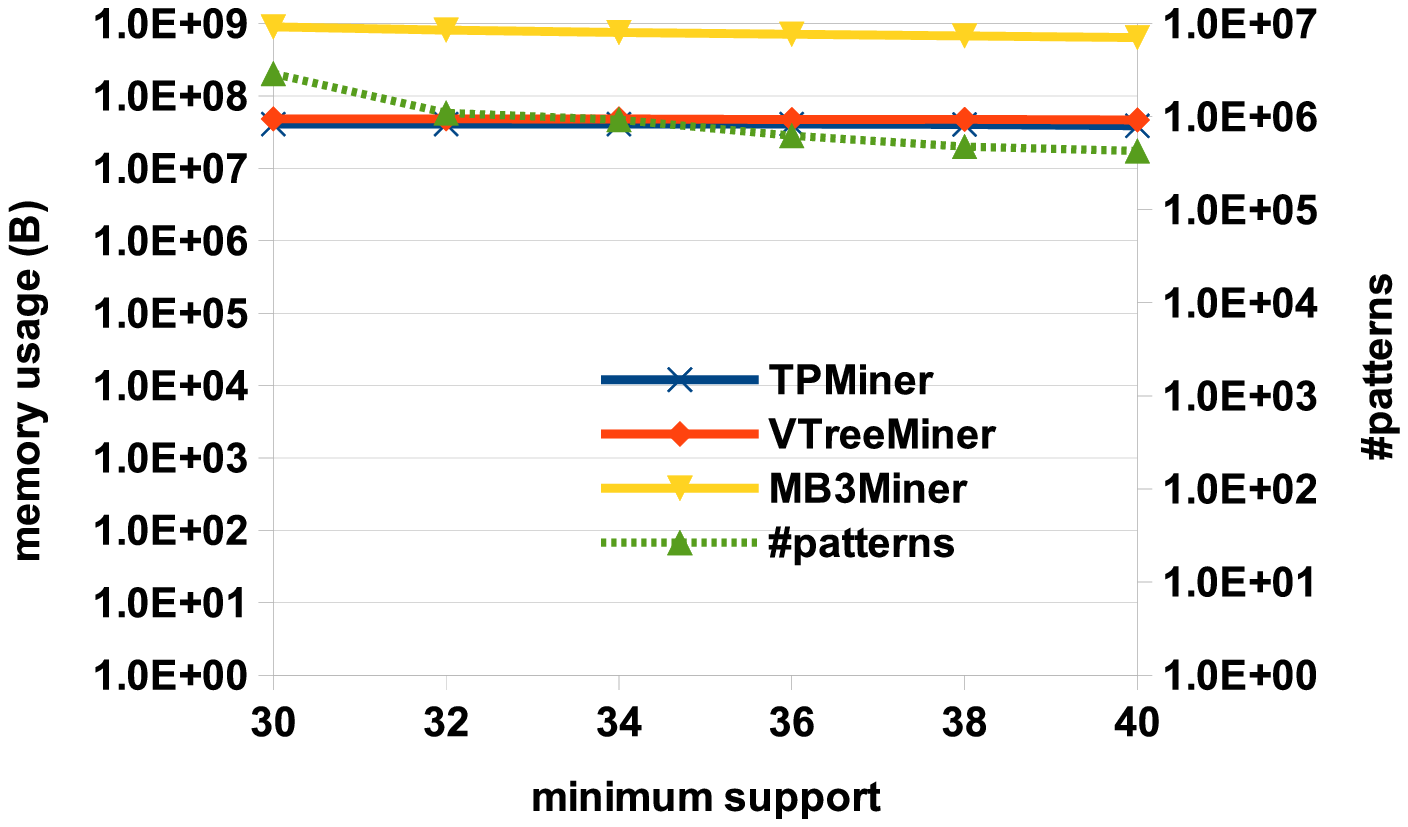}
\label{fig:m100000_time}
}
\caption
{
\label{fig:synthetic}
Comparison over synthetic datasets.
}
\end{figure*}

We also evaluated the efficiency of the proposed algorithm
on several synthetic datasets generated by the method described in \cite{36}.
The synthetic data generation program
mimics the web site browsing behavior of the user.
First a master web site browsing tree is built and then the
subtrees of the master tree are generated.
The program is adjusted by $5$ parameters:
($i$) the number of labels ($N$),
($ii$) the number of vertices in the master tree ($M$),
($iii$) the maximum fan-out of a vertex in the master tree ($F$),
($iv$) the maximum depth of the master tree ($D$), and
($v$) the total number of trees in the dataset ($T$).
Figure~\ref{fig:synthetic} compares the algorithms over the synthetic datasets.
The first synthetic dataset is D10 and uses the following
default values for the parameters: $N = 100$, $M=10,000$, $D=10$, $F=10$ and $T=100,000$.
Over this dataset, TPMiner is around $3$ times faster than VTreeMiner and VTreeMiner is slightly faster than MB3Miner.
The next synthetic dataset is D5, where $D$ is set to $5$ and for the other parameters, the default values are used.
Over this dataset, at $minsup=10$, $20$ and $30$, MB3Miner is aborted due to lack of memory.
We also evaluated the effect of $M$.
We set $M$ to $100,000$ and used the default values for the other parameters and generated the M100k dataset.
Over this dataset, TPMiner is faster than MB3Miner by a factor of $2$-$3$, and both TPMiner and MB3Miner are significantly faster than VTreeMiner.

\paragraph{Discussion.}
Our extensive experiments report that TPMiner always outperforms well-known existing algorithms.
Furthermore, there are several cases where TPMiner by order of magnitude is more efficient than any specific given algorithm. 
TPMiner and VTreeMiner require significantly less memory cells than MB3Miner.
This is due to the different large data-structures used by
MB3Miner such as the so-called EL, OC and VOL data structures and also
to the breadth-first search (BFS) strategy followed by MB3Miner \cite{23}.
Although TPMiner uses a more compact representation of
occurrences than VTreeMiner, this is hardly noticeable in the charts.
The reason is that the memory use is dominated by the storage of the frequent patterns.

In our experiments, we can distinguish two cases.
First, over datasets such as NASA and D5 (in particular for low values of $minsup$),
the Tree Model Guided technique used by MB3Miner does not significantly reduce the state space,
therefore, the algorithm fails or it does not terminate within a reasonable time. 
In such cases, TPMiner find all frequent patterns very effectively.
Second, over very large datasets (such as CSLOGS32241) or dense datasets (such as M100K)
where patterns have many occurrences, TPMiner becomes faster than VTreeMiner by order of magnitude.
This is due to the ability of TPMiner in frequency counting of patterns with many occurrences.
As discussed earlier, the \textbf{occ} data-structure used by TPMiner
can often represent and handle exponentially many occurrences with a single \textbf{occ} element,
while in VTreeMiner these occurrences are represented and handled one by one.

\section{Conclusion}
\label{section:conclusion}

In this paper, we proposed an efficient algorithm for subtree homeomorphism with application to frequent pattern mining.
We developed a compact data-structure, called \textbf{occ},
that effectively represents/encodes several occurrences of a tree pattern.
We then defined efficient join operations on \textbf{occ} that help us to
count occurrences of tree patterns according to occurrences of their proper subtrees.
Based on the proposed subtree homeomorphism method, we introduced TPMiner, an effective algorithm for finding
frequent tree patterns.
We evaluated the efficiency of TPMiner on several real-world and synthetic datasets.
Our extensive experiments show that TPMiner always outperforms well-known existing algorithms,
and there are several situations where the improvement compared to existing algorithms is significant.

\begin{acknowledgements}
We are grateful to Professor Mohammed Javeed Zaki for providing the VTreeMiner code,
the CSLOGS datasets and the TreeGenerator program,
to Dr Henry Tan for providing the MB3Miner code,
to Dr Fedja Hadzic for providing the Prions dataset
and to Professor Jun-Hong Cui for providing the NASA dataset.
Finally, we would like to thank Dr Morteza Haghir Chehreghani for his discussion and suggestions.
\end{acknowledgements}

\bibliographystyle{spmpsci}      
\bibliography{template}   

\end{document}